\documentclass[prx,aps,superscriptaddress,twocolumn,nofootinbib,longbibliography]{revtex4-2}

\usepackage{amsmath,amssymb,amsthm,bbm,bm,color,dsfont,float,graphicx,hyperref,makecell,mathrsfs,mathtools,nicefrac,pgfplots,physics,tikz,times,txfonts}
\usepackage{tikz-cd}
\usepackage[capitalise, noabbrev]{cleveref}
\usepackage[normalem]{ulem}
\usepackage{tcolorbox}
\hypersetup{colorlinks=true,linkcolor=purple,citecolor=purple,urlcolor=purple}

\pgfplotsset{compat=newest}
\tcbset{before skip=10pt,toptitle=2mm,bottomtitle=1mm,fonttitle=\bfseries}
\tcbuselibrary{theorems}
\tcbuselibrary{breakable}

\definecolor{NavyBlue}{rgb}{0.0, 0.0, 0.5}
\definecolor{OliveGreen}{rgb}{0.33, 0.42, 0.18}
\definecolor{def_color_frame}{RGB}{220,230,242}
\colorlet{def_color_back}{def_color_frame!30}
\definecolor{def_color_text}{RGB}{37,64,97}
\definecolor{def_color_frame2}{RGB}{242,200,200}
\colorlet{def_color_back2}{def_color_frame2!30}
\definecolor{def_color_text2}{RGB}{97,55,33}

\newtcolorbox[auto counter]{bluebox}[2][]{%
colback=def_color_back,colframe=def_color_frame,fonttitle=\bfseries,coltitle=def_color_text,float,floatplacement=t,title=Box~\thetcbcounter: #2,#1}

\newtcolorbox[use counter from=bluebox]{redbox}[2][]{%
colback=def_color_back2,colframe=def_color_frame2,fonttitle=\bfseries,coltitle=def_color_text2,float,floatplacement=t,title=Box~\thetcbcounter: #2,#1}

\def\E{ {\cal E} }

\def\P{ {\cal P} }
\def\Q{ {\cal Q} }

\def\>{\rangle}
\def\<{\langle}

\newcommand{\1}{\openone}

\newcommand{\poly}{\operatorname{poly}}

\renewcommand{\v}[1]{\ensuremath{\boldsymbol #1}}

\definecolor{bluecyan}{rgb}{0.27, 0.66, 0.88}
\definecolor{ppblue}{RGB}{46,117,182}
\definecolor{ppred}{RGB}{197, 90, 17}

\theoremstyle{plain}
\newtheorem{thm}{Theorem}

\newtheorem{lem}[thm]{Lemma}
\newtheorem{prop}[thm]{Proposition}

\newtheorem{defn}[thm]{Definition}

\usepackage{newfloat}
\usepackage[super]{nth}
\DeclareFloatingEnvironment[fileext=frm,placement={!ht},name=Box]{myfloat}

\DeclarePairedDelimiter\floor{\lfloor}{\rfloor}
\definecolor{beige}{rgb}{1.00,0.95,0.90}
\usepackage[framemethod=TikZ]{mdframed}

\begin{document}

\title{From the Unknown to the Desired:\protect\\
Transforming Unknown Initial States in the Resource Theory of Work and Heat}

\author{Tanmoy Biswas}
\email{tbiswas1@uni-koeln.de}
\affiliation{Department Mathematik/Informatik--Abteilung Informatik,\protect\\ Universit\"at zu K\"oln, Albertus-Magnus-Platz, 50923 K\"oln, Germany}
\affiliation{Theoretical Division (T-4), Los Alamos National Laboratory,\protect\\ Los Alamos, New Mexico 87545, USA.}

\author{Andreas Winter}
\email{andreas.winter@uni-koeln.de}
\affiliation{Department Mathematik/Informatik--Abteilung Informatik,\protect\\ Universit\"at zu K\"oln, Albertus-Magnus-Platz, 50923 K\"oln, Germany}
\affiliation{ICREA {\&} Grup d'Informaci\'{o} Qu\`{a}ntica, Departament de F\'{\i}sica,\protect\\ Universitat Aut\`{o}noma de Barcelona, 08193 Bellaterra (Barcelona), Spain}

\begin{abstract}
We study thermodynamic state transformations in the resource theory of work and heat when the initial quantum state is unknown. Moving beyond state-dependent protocols, we develop a universal framework in which a single transformation applies to all admissible input states. For mutually commuting conserved charges, we construct charge-conserving protocols whose transformation error, quantified by trace distance, decays super-polynomially with the number of system copies. As an application, we demonstrate universal work extraction: asymptotically optimal work can be extracted without microscopic knowledge of the initial state, with a super-polynomially small error. Our results show that thermodynamic state transformations remain achievable under limited prior information and establish a universal framework for resource-theoretic thermodynamics beyond the state-aware setting.
\end{abstract}

\date{8 September 2026}

\maketitle

\section{Introduction}

The resource-theoretic approach to quantum thermodynamics provides a powerful operational framework for analyzing state transformations under conservation laws, typically implemented through energy-conserving unitaries acting jointly on the system and a thermal environment 
\cite{Janzing2000, OppHorodecki3,Brandaosecondlaw,brandao2013resource}. This framework is especially effective when the initial state is known and thermal states are freely available through access to an infinite heat bath \cite{Seifert2008,Strasberg_Winter}. However, treating thermal states as free resources effectively presupposes access to an idealized, infinitely large reservoir, an assumption that need not hold in realistic settings \cite{Ergotropy, BeyondHalpern,Scharlau2018quantumhornslemma,Lobejko2021,vomEnde2022BathHamiltonians}. In many applications, including microscopic heat engines and information-processing devices, operations take place on timescales so short that the system must effectively be regarded as closed \cite{Aberg2013,Skrzypczyk2014, TajimaHayashi2017, ItoHayashi2018, Mohammady}. On such timescales, preparing thermal states can itself be highly nontrivial, especially at the low temperatures relevant for many quantum-computing platforms \cite{Strasberg_WinterPRA,Brown2011,Kempe,Bravyi2021PartitionFunctions}. Moreover, when the environment is finite, its state changes appreciably through interaction with the system \cite{Rivas2014NonMarkovianity,Breuer2009NonMarkovian}. A second limitation concerns knowledge of the initial state \cite{Jaynes1957InformationI,Jaynes1957InformationII,Landauer1991,Bennett1982}. Its preparation may involve a highly complex quantum process, such as a deep quantum circuit, or the system may be exposed to unknown noise, making an accurate microscopic description of the available state inaccessible \cite{HalpernPRAComplexity,Munson}. Consequently, the conventional assumption that the input state is completely known in advance---and that the transformation protocol can therefore be tailored to its microscopic details and to a desired target---need not be valid in realistic scenarios. These considerations motivate a fundamental question: \emph{To what extent are thermodynamic state transformations possible when the initial quantum state is unknown and access to an infinitely large thermal reservoir is unavailable?} This question is particularly important for work extraction, where the optimal transformation, and hence the maximum extractable work, generally depends on detailed knowledge of the initial state.

Ref.~\cite{Sparaciari-AET} addresses this problem for a \emph{known} initial state in the presence of a single conserved charge, namely the Hamiltonian, and without assuming access to an infinitely large thermal reservoir. In the thermodynamic limit, it shows that state transformations under energy-conserving unitaries are completely characterized by the entropy and average energy of the state. Remarkably, transforming \(n\) independent and identically distributed (i.i.d.) copies of the initial state requires only an ancillary system of sublinear size, equipped with a Hamiltonian whose norm grows sublinearly with \(n\). The resulting contributions of the ancilla to the total entropy and average energy are therefore \(o(n)\) and become negligible relative to the extensive system quantities in the thermodynamic limit. Ref.~\cite{Bera2019thermodynamicsas} extends this characterization to several mutually commuting conserved charges by reducing the problem to an effective single conserved quantity. A limitation of both constructions, however, is that implementing the required charge-conserving unitary relies on an additional coherence reservoir, or \emph{reference frame}, incorporated into the ancilla \cite{Gour_2008,PopescuWinter,AbergPRL}. Beyond its conceptual role, such a reference frame must be prepared in a suitable coherent state tailored to the transformation, which can be experimentally demanding \cite{Bartlett}. Moreover, the use of a finite reference frame leads to a transformation error that decreases only as a small power of the system size \(n\), reflecting a trade-off between the achievable accuracy and the size of the reference frame. This raises two important further questions: \emph{Can one realize the same class of thermodynamic state transformations without requiring a reference frame?} and \emph{Can one achieve a faster decay of the transformation error?}

Here, we address these questions from first principles by developing a framework for thermodynamic state transformations when the initial quantum state is unknown. Rather than designing a protocol tailored to a specific input state, we seek universal transformations that apply to any state within a prescribed class and produce a desired target state up to a controllably small error. This shifts the emphasis from state-dependent to universal protocols and brings typicality and concentration phenomena to the forefront in the asymptotic regime of many identical copies.

Our construction is inspired by the state-aware setting studied in Refs.~\cite{Sparaciari-AET,Bera2019thermodynamicsas}, where, in the absence of an infinite thermal reservoir, asymptotic state transformations are completely characterized by the entropy and the average values of the mutually commuting charges. In particular, equality of the entropy and average charges of the initial and target states constitutes the necessary and sufficient condition for asymptotic interconversion. We first develop a new protocol, distinct from that of Ref.~\cite{Sparaciari-AET}, for transforming \(n\) i.i.d.\ copies of a known initial state into a target state with the same entropy and average charges. The central ingredient is a \emph{charge--entropy normal form}: we show that \(n\) i.i.d.\ copies of a state can be mapped, in a charge-conserving manner and with an error that vanishes in the thermodynamic limit, to a tensor product of two states, one supported on a subspace of linear size and the other on a subspace of sublinear size. The linearly supported component, together with its associated charge content, depends only on the entropy and average charges of the state, whereas all microscopic state-dependent information is confined to the sublinear component. Consequently, states with the same entropy and average charges share the same linearly supported part of the charge--entropy normal form, and the desired transformation can be achieved by manipulating only the sublinear component. Importantly, our construction eliminates the need for an external reference frame and thereby yields a faster decay of the transformation error.

We then extend the construction to unknown initial states. Motivated by techniques from gentle tomography, we show that the relevant thermodynamic quantities---namely, the entropy and average charges---can be estimated while inducing only a vanishingly small disturbance to the state, using a coarse-grained, charge-preserving measurement. This effectively reduces the unknown-state problem to the known-state setting: once the entropy and average charges are estimated, the corresponding known-state transformation protocol can be applied. By combining this gentle estimation procedure with our reference-frame-free protocol, we obtain a universal scheme for asymptotic thermodynamic state transformations starting from an unknown quantum state.

As an application of our results, we consider universal work extraction in the absence of an infinitely large heat bath. Work extraction from unknown quantum states has previously been studied within thermal operations, where access to an infinitely large thermal reservoir is assumed \cite{WatanabePRL, WatanabeTakagi2026}. Standard optimal work-extraction protocols generally depend explicitly on knowledge of the initial state, with the maximum extractable work determined by the difference between its average energy and that of the Gibbs state with the same entropy \cite{Chakraborty2025SampleComplexity, Safranek, Canzio2025almostunknown,Biswas2026information,Pagliaro}. We show that, even when the initial state is unknown, one can construct a universal protocol that achieves optimal work extraction in the thermodynamic limit. Thus, asymptotically, the lack of knowledge of the initial state does not reduce the maximum extractable work compared with the case in which the state is completely known.

\section{Main Results}
\subsection{Setting}
Macroscopic thermodynamics provides a coarse-grained description of physical 
systems in terms of a small number of macroscopic quantities, such as entropy 
and the expectation values of energy etc, while disregarding the vast microscopic 
details of the underlying quantum state. Motivated by this perspective, it is natural to regard all quantum states sharing the same values of these macroscopic quantities as belonging to the same macroscopic equivalence class. 

For a set of mutually commuting charges $L_1,\ldots,L_t$ satisfying 
$[L_i,L_j]=0$, with $L_1=H$ denoting the Hamiltonian, we define
\begin{equation}
    \mathcal{C}(s,l_1,\ldots,l_t)
    :=
    \left\{
        \rho \,\middle|\,
        S(\rho)=s,\;
        \Tr(\rho L_i)=l_i,\quad i=1,\ldots,t
    \right\},
\end{equation}
where $\rho$ is the density matrix of a $d$-level quantum system, 
$\rho\in\mathcal{B}(\mathcal{H}_d)$. We refer to 
$\mathcal{C}(s,l_1,\ldots,l_t)$ as a \textit{charge-entropy equivalence class}. 
Such an equivalence class generally contains quantum states that may differ in their coherence, despite sharing the same macroscopic thermodynamic description. We are interested in investigating the operational significance of such equivalence classes in terms of state interconversion within the framework of resource theory of work and heat \cite{Sparaciari-AET}. In the following, we briefly review that framework.

The resource objects are quantum states on $n$ independent and identically 
prepared copies of this system, while the local 
dimension $d$ and the set of commuting charges $L_1,\ldots,L_t$ remain fixed. Then, the total charge 
is given as
\begin{equation}
L_c^{\times n}
=
\sum_{k=1}^{n}
\1^{\otimes(k-1)}
\otimes L_c
\otimes
\1^{\otimes(n-k)}
\end{equation}
for $c\in\{1,\ldots,t\}$. As we are interested in state interconversion between states within a charge-entropy equivalence class using $n$ copies of the state, we allow access 
to ancillary systems whose size and charges associated with it scale sublinearly with $n$. 
The contribution of such ancillas to the values of the charges and the entropy vanishes in the asymptotic limit, and therefore they do not changes the corresponding charge-entropy equivalence class. 

The allowed transformations are global unitaries $U$ satisfying
\begin{equation}
\label{defn_cc}
[U,L_c^{\times n}+L_{c,\mathrm{anc}}]=0,
\end{equation}
for every $c\in\{1,\ldots,t\}$ and $L_{c,\mathrm{anc}}$ denotes the charge associated with the ancilla system. We refer to such unitaries as \emph{charge-conserving unitaries}. Thus, one assumes complete coherent control over the joint system and ancilla, with the only fundamental constraint being the conservation of the charges. After the application of $U$, the ancillary system may be discarded by tracing it out. Since the ancilla is scaling sublinearly with the system size, this operation has a negligible asymptotic effect and therefore does not alter the corresponding charge-entropy equivalence class.

\begin{figure}[t]
    \centering
    \includegraphics[width=8 cm]{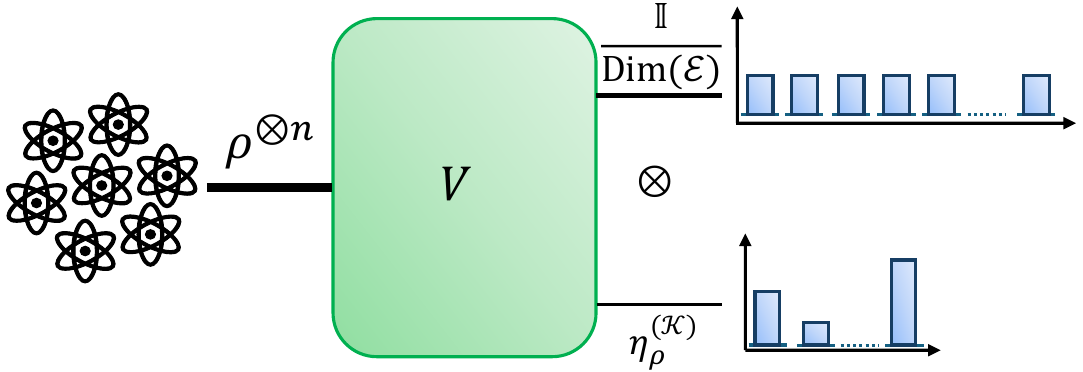}
    \caption{Conversion of an \(n\)-copy i.i.d.~state \(\rho^{\otimes n}\) into its \emph{charge--entropy normal form} by a unitary \(V\) that depends only on the charge--entropy equivalence class containing \(\rho\). The first component is maximally mixed, with its dimension and associated charges determined solely by the entropy and average charges of \(\rho\), whereas the second component retains the explicit dependence on \(\rho\), with its associated charge operators depending explicitly on the total charges \(L_c^{\times n}\).}
    \label{fig:state_transformation_1}
\end{figure}

\subsection{Transformation from known initial states}
In this section, we consider the transformation of known initial states. Specifically, we construct a charge-conserving unitary, depending only on the information of charge entropy equivalence class, that transforms independent and identically prepared copies of an initial state into a normal form consisting of a maximally mixed state tensored with a state depending on the initial state. We state the result in the following theorem.
\begin{thm}[charge-entropy normal form]
\label{Main_Central_Thm}
Let \(\rho\in\mathcal{C}(s,l_1,\ldots,l_t)\). For any \(\epsilon_n>0\), there exist Hilbert spaces
\(\mathcal{E}\) and \(\mathcal{K}\), associated with charges
\(L_c^{(\mathcal{E})}\) and \(L_c^{(\mathcal{K})}\), satisfying
\begin{align}
    \frac{1}{\poly(n)}e^{n(s-2\epsilon_n)}-1
    &\leq
    \dim (\mathcal{E})
    \leq e^{ns},
    &
    L_c^{(\mathcal{E})}
    &=nl_c\,\1^{(\mathcal{E})}\label{dimchargebd},
    \\
    \dim (\mathcal{K})
    &\leq \poly(n)e^{3n\epsilon_n},
    &
    \|L_c^{(\mathcal{K})}\|
    &\leq n\epsilon_n\label{dimchargebd2},
\end{align}
for \(c\in\{1,\ldots,t\}\). There exists a projector
\(P_{\epsilon_n}\) such that
\(\operatorname{supp}(P_{\epsilon_n})\simeq\mathcal{E}\otimes\mathcal{K}\),
and a unitary
\begin{equation}
V:\mathcal{H}_d^{\otimes n}=\operatorname{supp}(P_{\epsilon_n})\oplus\operatorname{supp}(P_{\epsilon_n})^\perp
\rightarrow
(\mathcal{E}\otimes\mathcal{K})
\oplus
\operatorname{supp}(P_{\epsilon_n})^\perp ,
\end{equation}
depending only on \(\mathcal{C}(s,l_1,\ldots,l_t)\), of the form
\begin{equation}\label{dsumstru}
V=U\oplus\1^{(\operatorname{supp}(P_{\epsilon_n})^\perp)}
\quad\mathrm{where}\quad
U:\operatorname{supp}(P_{\epsilon_n})
\rightarrow
\mathcal{E}\otimes\mathcal{K}.
\end{equation}
The unitary \(V\) acts non-trivially only on $\operatorname{supp}(P_{\epsilon_n})$ satisfies the charge-conservation condition
\begin{equation}\label{conscc}
\left(
L_c^{(\mathcal{E})}\otimes\1^{(\mathcal{K})}
+
\1^{(\mathcal{E})}\otimes L_c^{(\mathcal{K})}
\right)
\oplus
P_{\epsilon_n}^{\perp}L_c^{\times n}
=
V L_c^{\times n} V^{\dagger},
\end{equation}
for \(c\in\{1,\ldots,t\}\). Under this charge-preserving unitary \(V\), we have
\begin{equation}
\label{dt}
\left\|
V\rho^{\otimes n}V^{\dagger}
-
\frac{\1^{(\mathcal{E})}}
{\dim (\mathcal{E})}
\otimes
\eta_{\rho}^{(\mathcal{K})}
\right\|_1
\leq
\delta({\epsilon_n}),
\end{equation}
where 
\begin{equation}
\label{deltaepsilon_n}
    \delta({\epsilon_n}) 
    :=\poly(n) \exp\!\left[ -\frac{n}{4}
\left(
\frac{\epsilon_n}{\log(1/\epsilon_n)}
\right)^2
\right],
\end{equation}
and $\1^{(\mathcal{E})}$ denotes identity operator on Hilbert space $\mathcal{E}$, and $\eta_{\rho}^{(\mathcal{K})}$ is a density matrix on $\mathcal{K}$.
\end{thm} 

The full proof of this theorem is given in Appendix~\ref{cc1_calc}, while a proof sketch containing the main ideas is presented in Subsection A of the Methods. 
%

We refer to the state
\begin{equation}
    \tau^{(\mathcal E)}\otimes
    \eta_{\rho}^{(\mathcal K)},
    \qquad
    \tau^{(\mathcal E)}
    :=
    \frac{\1^{(\mathcal E)}}{\dim (\mathcal E)},
\end{equation}
equipped with the charges $L_c^{(\mathcal E)}$ and $L_c^{(\mathcal K)}$, for 
$c\in\{1,\ldots,t\}$, acting on the Hilbert spaces $\mathcal E$ and $\mathcal K$, respectively, as the \emph{charge-entropy normal form} of the state $\rho$ (see Fig.~\ref{fig:state_transformation_1}).

Let us discuss the structure of the charge-entropy normal form. This structure is determined predominantly by the charge-entropy equivalence class $\mathcal{C}(s,l_1,\ldots,l_t)$ containing $\rho$. The only dependence of the normal form on the particular state $\rho$ appears in the state of the second subsystem $\eta_{\rho}^{(\mathcal K)}$, whose dimension is bounded as given in Eq.~\eqref{dimchargebd}. Therefore, all states belonging to the same charge-entropy equivalence class admit the similar charge-entropy normal form structure, differing only in the state of the subsystem $\eta_{\rho}^{(\mathcal K)}$.

The Hilbert space $\mathcal E$ carries only the trivial charge [see Eq.~\eqref{dimchargebd}], whereas all non-trivial charges are associated with the Hilbert space $\mathcal K$. To ensure charge conservation, these charges are constructed such that their corresponding norms are bounded by $n\epsilon_n$ [see Eq.~\eqref{dimchargebd2}]. Moreover, the charges associated with $\mathcal E$ and $\mathcal K$ are constructed solely from the charge-entropy equivalence class [see Eqs.~\eqref{dimchargebd} and \eqref{conanc}]. Therefore, every state within the same charge-entropy equivalence class admits a charge-entropy normal form with the same charge structure on the subsystems $\mathcal E$ and $\mathcal K$. This observation provides the foundation for characterizing state interconversion within a charge-entropy equivalence class, as formalized in the following theorem.
\begin{thm}[Asymptotic equivalence theorem]
\label{thm2}
Let $\rho,\sigma\in\mathcal{C}(s,l_1,\ldots,l_t)$ be any two states
associated with the commuting charges $\{L_c\}_{c=1}^t$. For any
$\epsilon_n>0$, there exist an ancilla state
$\eta_{\mathrm{anc}}\in\mathcal{B}(\mathcal{H}_{\mathrm{anc}})$ with charge
operator $L_{c,\mathrm{anc}}$, and a unitary $U$ satisfying
\begin{equation}
[U,L_c^{\times n}+L_{c,\mathrm{anc}}]=0,
\end{equation}
such that transformation error satisfy
\begin{align}\label{T_err_main}
&\left\|
\Tr_{\mathrm{anc}}
\left[
U
(\rho^{\otimes n}\otimes\eta_{\mathrm{anc}})
U^{\dagger}
\right]
-\sigma^{\otimes n}
\right\|_1
\lesssim 2\delta({\epsilon_n}),
\end{align}
where $\delta({\epsilon_n})$ is given in Eq.~\eqref{deltaepsilon_n}. Moreover, the dimension of the ancilla system and the corresponding charges satisfy
\begin{equation}
\label{diancbd_chargebd}
\dim (\mathcal{H}_{\mathrm{anc}})
\leq
\poly(n)\exp\!\left(3n\epsilon_n\right),
\qquad
\|L_{c,\mathrm{anc}}\|\leq n\epsilon_n .
\end{equation}
\end{thm}
\begin{proof}
By Theorem~\ref{Main_Central_Thm}, there exists a charge-conserving unitary \(V\) associated with the charge-entropy equivalence class \(\mathcal C(s,l_1,\ldots,l_t)\) such that \(\rho^{\otimes n}\) and \(\sigma^{\otimes n}\) are mapped, up to an error \(\delta({\epsilon_n})\) [cf.~Eq.~\eqref{dt}], to their respective charge-entropy normal forms,
\begin{equation}
\tau^{(\mathcal E)}\otimes\eta_{\rho}^{(\mathcal K)},
\qquad
\tau^{(\mathcal E)}\otimes\eta_{\sigma}^{(\mathcal K)}.
\end{equation}
Now, we choose the ancillary system along with the charges to be
\begin{equation}
\mathcal H_{\mathrm{anc}}=\mathcal K,
\quad
\eta_{\mathrm{anc}}=\eta_{\sigma}^{(\mathcal K)},\quad L_{c,\mathrm{anc}}=L_c^{(\mathcal{K})},
\end{equation}
where $L_c^{(\mathcal{K})}$ is the charge from Theorem \ref{Main_Central_Thm} with $c\in\{1,\ldots,t\}$. The desired charge-conserving unitary is defined as
\begin{equation}\label{3product}
U = V^{\dagger}
\,\mathrm{SWAP}^{(\mathcal{K})}
\,V,
\end{equation}
where \(\mathrm{SWAP}^{(\mathcal{K})}\) exchanges the subsystems $(\eta_{\rho}^{(\mathcal K)}$ and $\eta_{\sigma}^{(\mathcal K)})$ supported on Hilbert space $\mathcal{K},$ and acts trivially on the rest. Since the charges associated with \(\mathcal K\) depend only on the charge-entropy equivalence class, the two subsystems carry identical charges. Hence the $\mathrm{SWAP}^{(\mathcal{K})}$ is charge conserving, and therefore so is \(U\). The action of \(U\) is illustrated in Fig. \ref{fig:state_transformation_2}

Since the SWAP operation is exact, the only errors arise from the applications of \(V\) and \(V^\dagger\). By the triangle inequality, the total transformation error is therefore at most $2\delta({\epsilon_n})$ where $\delta({\epsilon_n})$ is given in Eq.~\eqref{deltaepsilon_n}, hence establishing Eq.~\eqref{T_err_main}. Finally, the bounds on the ancilla dimension and charges in Eq.~\eqref{diancbd_chargebd} follow immediately from Eqs.~\eqref{dimchargebd} and \eqref{dimchargebd2} of Theorem~\ref{Main_Central_Thm}.
\end{proof}

\begin{figure}[t]
    \centering
    \includegraphics[width=8.5 cm]{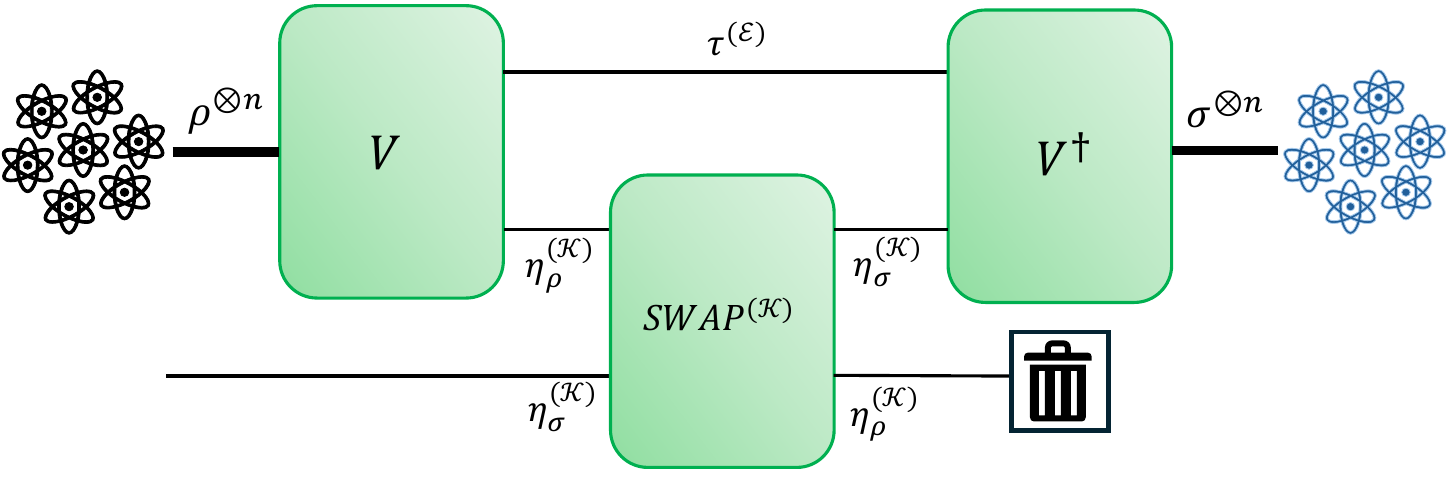}
    \caption{Transformation of \(\rho^{\otimes n}\) into \(\sigma^{\otimes n}\), where \(\rho\) and \(\sigma\) belong to the same charge--entropy equivalence class, assisted by a sublinear-sized ancilla. The unitary \(V\) maps \(\rho^{\otimes n}\) to the charge--entropy normal form \(\tau^{(\mathcal{E})}\otimes\eta^{(\mathcal{K})}_{\rho}\) with error at most \(\delta(\epsilon_n)\). Since the linearly supported component \(\tau^{(\mathcal{E})}\) is identical for states within the same charge--entropy equivalence class, the desired transformation is achieved by swapping \(\eta^{(\mathcal{K})}_{\rho}\) with an ancilla prepared in the state \(\eta^{(\mathcal{K})}_{\sigma}\), followed by \(V^\dagger\). The total transformation error is at most \(2\delta(\epsilon_n)\), accounting for the errors incurred in applying \(V\) and \(V^\dagger\).}
    \label{fig:state_transformation_2}
\end{figure}
In the following, we show that the state transformation can be implemented with a super-polynomially small error while requiring only a sublinear-sized ancilla system. 

\begin{thm}[Super-polynomial accuracy with sublinear ancilla]
\label{corr1}
Under the assumptions of Theorem~\ref{thm2}, choosing
\begin{equation}
    \epsilon_n=\frac{(\log n)^2}{\sqrt{n}},
\end{equation}
the charge-conserving unitary $U$ can be implemented with an ancillary system satisfying
\begin{align}
    \log\left(\dim (\mathcal{H}_{\mathrm{anc}})\right)
    &=O\left(\sqrt{n}(\log n)^2\right),\\
    \|L_{c,\mathrm{anc}}\|
    &\leq\sqrt{n}\left(\log n\right)^2,
\end{align}
while the transformation error decays super-polynomially as
\begin{equation}
    \exp\bigl[-\Theta((\log n)^2)\bigr].
\end{equation}
\end{thm}
The proof of this theorem can be found in Subsection B of the Methods.

At this point, we make an important remark. The results stated in Theorems~\ref{thm2} and~\ref{corr1} resemble the asymptotic equivalence established in~\cite[Thm.~1]{Sparaciari-AET}. However, our proof is entirely different. Besides providing a simpler derivation and extending the result to arbitrary sets of commuting charges, we show that the transformation can be achieved with only a sublinear-sized ancilla while the error decreases \emph{super-polynomially} with the number of copies. In contrast, the proof of~\cite{Sparaciari-AET} achieves only \emph{polynomial} error decay under the same ancilla-size constraint. The intuition behind this is explained in the following paragraph.

While our framework considers a general set of commuting charges, \cite{Sparaciari-AET} focuses on the special case of a single charge, the Hamiltonian, and employs a sublinear-sized product ancilla of the form $\eta=\eta_1\otimes\eta_2$. The state \begin{equation} \eta_1:= \left(\frac{\1}{d}\right)^{\otimes n_1} \otimes \ketbra{0}^{\otimes n_2} \end{equation} serves as an information battery with zero Hamiltonian, enabling the modification of the spectrum of $\rho^{\otimes n}$ to that of $\sigma^{\otimes n}$ through a global unitary transformation. On the other hand, the state \begin{equation}\label{rf_sparaciari}
\eta_2=\ketbra{\psi}, \qquad \ket{\psi} = \frac{1}{\sqrt{|k\mathcal{L}+\mathcal{L}|}} \sum_{h\in k\mathcal{L}+\mathcal{L}} \ket{h}, 
\end{equation} 
acts as a reference frame with Hamiltonian having operator norm $o(n)$, allowing the aforementioned unitary transformation to be lifted to an energy-conserving unitary. Here,
\begin{equation}
    \mathcal{L}
    =
    \mathcal{E}_{\rho}
    \cup
    (-\mathcal{E}_{\rho})
    \cup
    \{0\},
\end{equation}
where $\mathcal{E}_{\rho}$ denotes the set of typical energy values of $H^{\times n}$, so that $|\mathcal{L}|=\poly(n)$.

We emphasize that implementing an arbitrary unitary by means of a reference frame necessarily involves a trade-off between the ancilla size and the implementation error. A larger reference frame yields a more accurate energy-conserving implementation, but at the cost of a larger ancilla. Indeed, \cite[Lemma~16]{Sparaciari-AET} and more specifically \cite[Eq.~(A16)]{Sparaciari-AET} shows that the transformation error is bounded from above by a constant multiple of 
\begin{equation}
  \label{sparaciari_lb2}
    \frac{|k\mathcal{L}+\mathcal{L}|}{|k\mathcal{L}|}-1,
\end{equation}
and inspection of the calculation on the functioning shows that the error is also bounded from below by a power of this expression. 
Since $|\mathcal{L}|=\poly(n)$, one can ensure sublinear ancilla scaling by suitably restricting the choice of $k$ in Eq.~\eqref{rf_sparaciari} (e.g., $k\leq n^{1/7}$ as in \cite{Sparaciari-AET}). However, the same parameter $k$ also constrains the achievable transformation error. Consequently, the lower bound on the transformation error in Eq.~\eqref{sparaciari_lb2} exhibits at best polynomial decay in $n$.  

Our proof avoids this bottleneck altogether. Rather than relying on a generic reference frame construction for implementing arbitrary unitaries, we explicitly construct the required charge-conserving unitary via the charge entropy normal form, thereby achieving super-polynomially small transformation error with the same sublinear ancilla scaling.

\subsection{Transformation from unknown initial states}
\begin{figure*}[!htbp]
    \centering
    \includegraphics[width=13 cm]{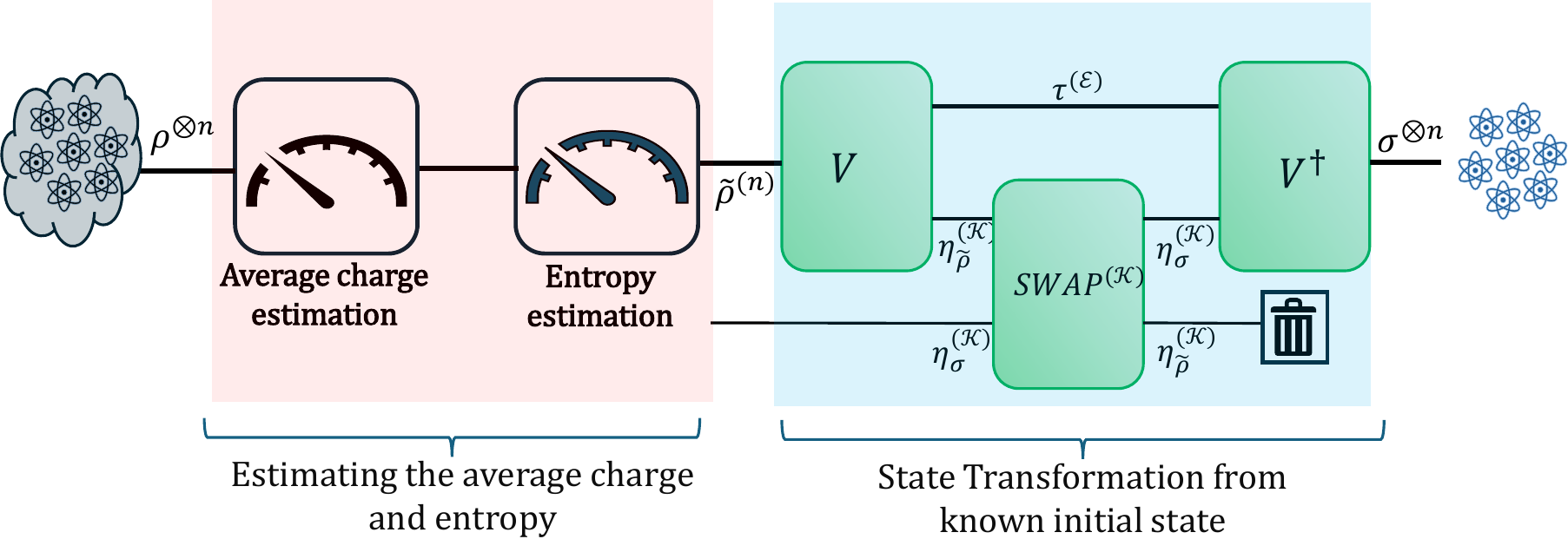}
    \caption{Schematic illustration of universal thermodynamic state transformations from unknown quantum states. Given \(n\) copies of an unknown initial state \(\rho^{\otimes n}\), one first performs a charge-preserving measurement to estimate the entropy and average charges. Once these quantities are determined, one applies a state-transformation protocol that depends only on the estimated entropy and average charges. Combining these two steps---estimation of the entropy and average charges, followed by the corresponding state-transformation protocol---yields a universal transformation protocol.}
    \label{fig:work_extraction_unknown_22}
\end{figure*}

In this section, we consider the problem of transforming an unknown initial state
$\rho\in\mathcal{B}(\mathcal{H}_d)$ into a given target state
$\sigma\in\mathcal{B}(\mathcal{H}_d)$ via a charge conserving unitary $U$. We assume that no prior information
about the initial state is available apart from the dimension $d$ of the Hilbert
space. In particular, neither the entropy nor the average charge of $\rho$ is
known, and consequently, its charge-entropy equivalence class is also unknown. Therefore, to determine whether the transformation from $\rho$ to the target
state $\sigma$ is achievable via a charge-conserving unitary in presence of sublinear sized ancilla with sublinear charge, one needs to
verify whether $\rho$ and $\sigma$ belong to the same charge-entropy
equivalence class. If they belong to the same charge-entropy equivalence class,
then Theorem~\ref{thm2} can be applied to construct such a charge-conserving unitary
$U$ that implements the transformation. Otherwise, the state transformation is not possible.

To determine the entropy and average charge of the initial state $\rho$, one
first performs measurements on the $n$-copy state $\rho^{\otimes n}$ while
ensuring negligible disturbance. Violating this constraint may result
in a post-measurement state belonging to a different charge-entropy equivalence class, thereby preventing the subsequent state transformation. 

Since the allowed operations are restricted to charge-conserving unitaries [see Eq.~\eqref{defn_cc}]
acting on the system and ancilla, we consider only those measurements that admit
a realization within this class of operations. Specifically, we extend the ancilla by introducing a probe system
of sublinear size, initialized in a pure state \(\ket{0}\) and carrying zero
charge for all $c\in\{1,\ldots,t\}$ i.e.,

\begin{equation}
    \forall c\in\{1,\ldots,t\}\quad L_{c,\text{probe}}=0.
\end{equation}
The measurement procedure is then implemented by applying a
charge-conserving unitary on the joint system consisting of
\(\rho^{\otimes n}\), the original ancilla, and the probe, followed by a
projective measurement on the probe. In this way, the measurement process is
realized entirely within the class of charge-conserving operations, with the
probe acting as a register that stores the measurement outcome.

By the Wigner-Araki-Yanase (WAY) theorem, repeatable projective measurements
compatible with the conservation laws must have measurement operators that
commute with the conserved quantities; see \cite{AhmadiJenningsRudolph2013} for a modern exposition. Therefore, the allowed projective
measurements on \(\rho^{\otimes n}\) are restricted to those with projectors
\(\{M_i\}\) satisfying
\begin{equation}
    \forall i\;\forall c\in\{1,\ldots,t\} \quad
    [M_i,L_c^{\times n}]=0.
\end{equation}
We refer to such measurements as \emph{charge-preserving measurements}.

As discussed earlier, unlike the case of a known initial state, here the average charge and the entropy of the initial state are unknown. Our objective is therefore to estimate these quantities efficiently by performing charge-preserving measurements that induce only negligible disturbance to the state, thereby preserving its suitability for the subsequent state transformation. Once the entropy and average charges of the initial state have been estimated, the problem reduces to the known-state setting, and the state transformation protocol described in the previous section can be applied directly. 

We now explain how to efficiently estimate the average conserved charges, namely the expectation values $\Tr(L_c\rho)$ for $c\in\{1,\ldots,t\}$, while causing an asymptotically vanishing disturbance to the state. Since the charges mutually commute, they admit a common eigenbasis $\{\ket{x}\}_{x=1}^{d}$ such that
\begin{equation}
\forall c\in\{1,\ldots,t\}, \qquad
L_c=\sum_{x=1}^{d}L_c(x)\ket{x}\!\bra{x}.
\end{equation}
Consequently,
\begin{equation}
\Tr(L_c\rho)=\sum_{x=1}^{d}L_c(x)\langle x|\rho|x\rangle.
\end{equation}
Therefore, it suffices to estimate the diagonal elements of $\rho$ in the common eigenbasis. Defining
\begin{equation}
\alpha_x:=\langle x|\rho|x\rangle,\qquad x\in\{1,\ldots,d\},
\end{equation}
the expectation value of every conserved charge is recovered as
\begin{equation}
\Tr(L_c\rho)=\sum_{x=1}^{d}L_c(x)\alpha_x.
\end{equation}
To estimate the values of $\alpha_x$, we employ the gentle tomography protocol introduced in \cite{BHL2006_PRA}, that uses a coarse-grained projective measurement on the unknown state $\rho^{\otimes n}$ to learn the state while causing only an asymptotically vanishing disturbance.

\begin{thm}[Estimating average charges with charge-preserving measurement]
\label{thm_charge_estimation}
   For \(m=n^z\) with \(0<z<\frac12\), there exists a charge preserving measurement for estimating \(\alpha_x\) has failure probability at most $\mathcal{O}\!\left(n^{\,z-\frac12}\log n\right).$ Upon successful execution, the estimated value of $\alpha_x$ deviates from true value of $\alpha_x$ by at most $ \mathcal{O}\!\left(n^{-z} \log n\right),$ while the disturbance induced on the state is upper-bounded by $\mathcal{O}\!\left(n^{-p}\right)$ where $p$ is a positive constant.
\end{thm}

We now present the theorem establishing that the von Neumann entropy of the initial state can be estimated efficiently using charge preserving measurement inducing only a negligible disturbance to the state.  

\begin{thm}[Estimating entropy with charge-preserving measurement]
\label{thm_entropy_estimation}
For any \(0<z<\tfrac{1}{2}\), there exists a charge-preserving measurement applied to \(\rho^{\otimes n}\) that estimates the von Neumann entropy \(S(\rho)\) of an unknown quantum state \(\rho\). The protocol fails with probability at most $\mathcal{O}\!\left(n^{\,z-\frac{1}{2}}(\log n)^{3/2}\right).$ Conditioned on success, the estimated entropy deviates from the true value by at most $\mathcal{O}\!\left(n^{-z}\log n\right)$, while the disturbance induced on the state is bounded by $\mathcal{O}\!\left(n^{-q}\right)$, where $q>0$ is a constant.
\end{thm}

We emphasize that, in our protocol, the charge-preserving measurement used to estimate the entropy is applied to the post-measurement state obtained after estimating the average energy (see Fig.~\ref{fig:work_extraction_unknown_22}), whereas Theorem~\ref{thm_entropy_estimation} establishes the entropy-estimation guarantee when the same measurement is applied directly to the initial state $\rho^{\otimes n}$. Nevertheless, by the triangle inequality, the total disturbance resulting from the sequential estimation of the average charge and the entropy is bounded by the sum of the individual disturbances associated with the two estimation procedures. In particular, if the disturbances associated with the charge and entropy estimation are $\delta_L$ and $\delta_S$, respectively, then
\begin{equation}
\delta_{\mathrm{tot}}
\leq \delta_L+\delta_S.
\end{equation}
As Theorems \ref{thm_charge_estimation} and \ref{thm_entropy_estimation} give $\delta_L\leq\mathcal{O}(n^{-p})$ and $\delta_S\leq \mathcal{O}(n^{-q})$, the total disturbance remains asymptotically vanishing,
\begin{equation}\label{addisturbance}
\delta_{\mathrm{tot}}
\leq\mathcal{O}\!\left(n^{-\min\{p,q\}}\right).
\end{equation}

\section{Application: Universal work extraction}
We now use our results to show that the maximal amount of work extractable from an unknown initial state is asymptotically the same as when the initial state is known. To this end, let us first briefly recall work extraction from a known initial state within the resource theory of work and heat, where the goal is to extract the maximum amount of work from many copies of a given state. To describe work extraction, we focus on the setting in which the Hamiltonian is the only conserved charge. To model the extracted work explicitly, we introduce a battery consisting of $l$ copies of a system initially prepared in its pure ground state, $\ketbra{E_{\min}}^{\otimes l}$, where $l$ is to be determined. Extracting work from a known state $\rho^{\otimes n}$ then amounts to implementing the transformation
\begin{equation}
    \omega_{\mathrm{in}}
    =
    \rho^{\otimes n}\otimes\ketbra{E_{\min}}^{\otimes l}
\end{equation}
into
\begin{equation}
    \omega_{\mathrm{tar}}
    =
    \sigma^{\otimes n}\otimes\ketbra{E_{\text{tar}}}^{\otimes l},
\end{equation}
where the final state $\sigma$ is chosen so as to maximize the extracted work per copy of $\rho$,
\begin{equation}
    W:=\frac{l}{n}\left(E_{\text{tar}}-E_{\min}\right).
\end{equation}
By Theorem~2, such a transformation is possible if and only if $\omega_{\mathrm{in}}$ and $\omega_{\mathrm{tar}}$ have the same energy and entropy. Since the battery is pure both initially and finally, entropy conservation implies
\begin{equation}
    S(\omega_{\mathrm{in}})=S(\omega_{\mathrm{tar}})\quad\Leftrightarrow\quad S(\rho)=S(\sigma).
\end{equation}
Moreover, energy conservation, $E(\omega_{\mathrm{in}})=E(\omega_{\mathrm{tar}})$
\begin{equation}\label{Ext_work}
    W=\frac{l}{n}\left(E_{\text{tar}}-E_{\min}\right)
    =E(\rho)-E(\sigma).
\end{equation}
Hence, maximizing the energy gained by the battery, or equivalently maximizing $l$, reduces to minimizing the average energy of $\sigma$ subject to the entropy constraint $S(\sigma)=S(\rho)$. As shown in Ref.~\cite{Sparaciari-AET}, the state that minimizes the energy at fixed von Neumann entropy is the Gibbs state $\tau_{\rho}$ satisfying
\begin{equation}
    \tau_{\rho}
    =
    \frac{e^{-\beta H}}{\Tr\!\left(e^{-\beta H}\right)},
    \qquad
    S(\tau_{\rho})=S(\rho),
\end{equation}
where $\beta$ is chosen such that $\tau_{\rho}$ has the same entropy as $\rho$ which implies to obtain maximum work the target state should be
\begin{equation}
    \omega_{\text{fin}}:=\tau_{\rho}^{\otimes n} \otimes\ketbra{E_{\max}}^{\otimes l}.
\end{equation}
Thus, the maximal extractable work is
\begin{equation}\label{maxWork}
    W_{\max}=E(\rho)-E(\tau_{\rho}),
\end{equation}
as illustrated by the energy--entropy diagram in Fig.~\ref{fig:work_extraction}. It is worth emphasizing that the expression in Eq.~\eqref{maxWork} is consistent with the results of Refs.~\cite{AlickiBoost,Ergotropy}.

This characterization reveals an important fact: for a known initial state, the maximal extractable work depends only on two quantities, namely its average energy and von Neumann entropy. This observation allows us to extend the same optimal work-extraction to an unknown initial state. The protocol proceeds as follows:
\begin{enumerate}
    \item We first estimate the average energy and entropy of the unknown state. By Theorems~\ref{thm_charge_estimation} and~\ref{thm_entropy_estimation}, these quantities can be estimated using an energy-preserving measurement. Moreover, by Eq.~\eqref{addisturbance}, the disturbance induced by this measurement vanishes asymptotically. By continuity of both the average energy and the von Neumann entropy, the corresponding changes in these quantities therefore also vanish in the asymptotic limit.

    \item Having obtained asymptotically accurate estimates of the average energy and entropy, we then apply the optimal work-extraction protocol for the known-state setting given in Theorem~\ref{thm2}. Since the disturbance generated by the estimation step vanishes asymptotically, the resulting changes in average energy and entropy also vanish. Consequently, the work extracted according to Eq.~\eqref{Ext_work} converges to the same value as in the case where the initial state is known.
\end{enumerate}
Therefore, in the asymptotic limit, ignorance of the initial state does not reduce the maximal amount of extractable work: the same optimal work-extraction can be achieved by first estimating the relevant thermodynamic quantities through an asymptotically non-disturbing, energy-preserving measurement and subsequently applying the optimal work extraction protocol from known initial state.

\begin{figure}[t]
    \centering
    \includegraphics[width=8.5 cm]{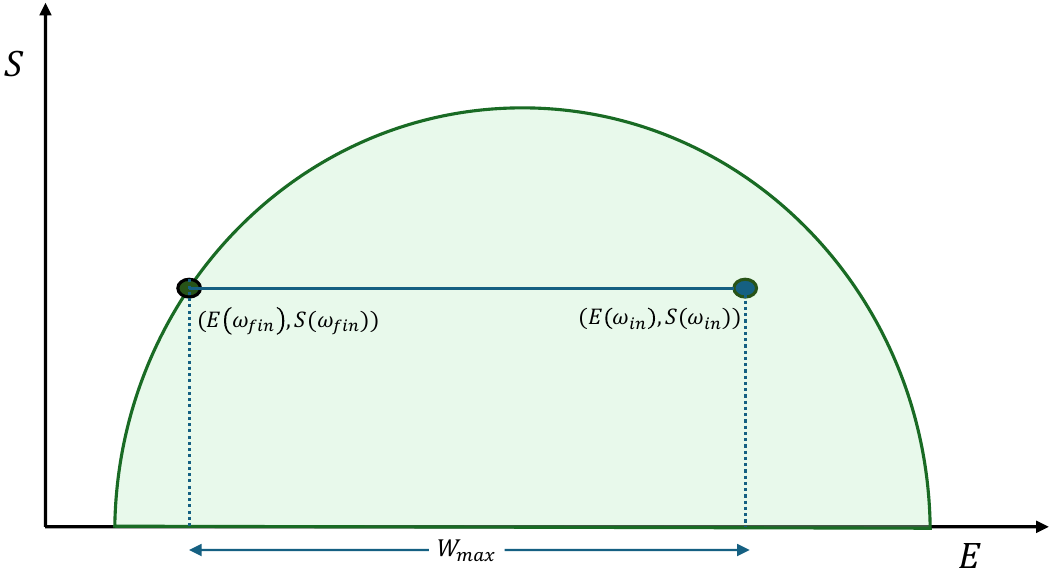}
    \caption{Energy--entropy diagram illustrating the maximum extractable work per copy from \(\rho\). The diagram shows that, to extract the maximum amount of work from \(\rho\), one must transform it into the Gibbs state \(\tau_{\rho}\) that has the same entropy as \(\rho\) and is defined with respect to the same Hamiltonian. }
    \label{fig:work_extraction}
\end{figure}

\section{Discussion and outlook}
In this work, we have developed a general framework for thermodynamic state transformations starting from an unknown quantum state within the resource theory of work and heat, where the allowed operations consist of charge-conserving unitaries assisted by a sublinear-sized ancilla (in terms of both entropy, energy and other commuting charges).

As a first step, we have introduced a protocol for state interconversion that transforms \(n\) i.i.d.~copies of a known initial state into a desired target state by exploiting the permutation invariance of i.i.d.~states and of the corresponding total-charge operators. We showed that the protocol extends naturally to multiple commuting conserved charges and achieves a super-polynomially vanishing transformation error in the asymptotic limit. Importantly, we identified that our construction does not require a reference frame as part of the ancillary system to enforce charge conservation, allowing this super-polynomial convergence to be achieved using only a sublinear-sized ancilla. Our results further establish that the entropy and average conserved charges determine the equivalence classes within which state interconversion is possible in the thermodynamic limit, \(n\rightarrow\infty\). In this sense, our framework provides an explicit passage from microscopic quantum states to thermodynamic macrostates: while the underlying microstates may differ, their asymptotic interconversion properties are determined solely by macroscopic quantities, namely entropy and average charges.

We then extended the framework to the \emph{universal} setting in which the initial state is unknown. To this end, we employed gentle tomography based on charge-preserving measurements performed on the $n$ i.i.d.~copies of the initial state to estimate its entropy and average energy -- the key thermodynamic quantities required to implement the transformation -- while inducing only asymptotically negligible disturbance. We thereby showed that an asymptotic transformation to a target state is possible if and only if the target state has the same entropy and average energy as the initial state. Once these quantities are estimated, the transformation can be implemented by applying the corresponding protocol developed for a known initial state; otherwise, such a transformation is impossible. As a concrete thermodynamic application, we showed that optimal work extraction remains achievable even without knowing the initial state. Taken together, these results establish a general operational framework for asymptotic thermodynamic state transformations under conservation laws, demonstrating that ignorance of the microscopic initial state need not constitute a fundamental limitation.

Our work opens up several interesting directions for future research. First, while our results are established in the i.i.d.~setting, it is natural to ask whether they can be extended to more general \(n\)-copy states, in particular permutation-invariant and almost-i.i.d.~states. 
Schur-Weyl duality, which plays a central role in our proof of Theorem~1, applies naturally to permutation-invariant states, suggesting that an analogous charge-entropy normal form may exist beyond the i.i.d.~regime. 
In a different direction, we might want to extend the results to general tensor product states, as was shown in \cite{KhanianBeraRieraLW2022}.

Second, it would be interesting to investigate finite-size corrections to our results \cite{Tomamichel2016finite,ChubbPRA}. Throughout this work, we operate in the large-deviation regime, where fluctuations of entropy and conserved charges become negligible relative to their extensive values \cite{Chubb2018beyondthermodynamic,BiswasPRE}. Going beyond this thermodynamic limit through a moderate-deviation analysis may reveal a refined characterization of state transformations involving not only entropy and average charges, but also their fluctuations, such as the corresponding variances. 

Third, our analysis is restricted to mutually commuting conserved charges. A complete characterization of state transformations in the presence of non-commuting charges remains open. Previous approaches have exploited a seminal result of Ogata \cite{Ogata:non-commuting-to-commuting} to approximate the collective non-commuting charges of many copies by a set of commuting observables (see \cite{KhanianBeraRieraLW2022} and \cite{Yunger-Halpern2016}), and crucially by allowing approximately charge-conserving unitaries. It remains unclear, however, how the transformation structure changes when the non-commuting charges have to be conserved exactly, i.e.~without allowing asymptotically vanishing charge-conservation violation, and how one should characterize transformations within the corresponding charge-entropy equivalence classes.

Finally, an important direction is to determine whether the state-transformation protocol underlying Theorem~1 can be implemented experimentally, thereby connecting the present asymptotic framework with realistic physical platforms.

\bigskip\noindent
\textbf{Acknowledgments.}
TB thanks Andrea Canzio and Sayantan Chakroborty for useful discussions. The authors are supported by the Alexander von Humboldt Foundation.
TB was additionally supported by the U.S. Department of Energy, Office of Science, Accelerated Research in Quantum Computing, Fundamental Algorithmic Research toward Quantum Utility (FAR-Qu). 
AW acknowledges furthermore support by the European Commission QuantERA project ExTRaQT (Spanish MICIN grant no.~PCI2022-132965), by the Spanish MICIN (project PID2022-141283NB-I00) with the support of FEDER funds, by the Spanish MICIN with funding from European Union NextGenerationEU (PRTR-C17.I1) and the Generalitat de Catalunya, by the Spanish MTDFP through the QUANTUM ENIA project: Quantum Spain, funded by the European Union NextGenerationEU within the framework of the ``Digital Spain 2026 Agenda'', and by the Institute for Advanced Study of the Technical University Munich.

\bibliography{main}


\section*{Methods}
\subsection{Proof sketch of Theorem \ref{Main_Central_Thm}}
Employing the permutation invariance of the state $\rho^{\otimes n}$ and the total charges $L_c^{\times n}$, Schur-Weyl duality \cite{goodman_wallach_representation_theory} yields the decompositions
\begin{equation}\label{Tensorstate_totalcharge_main}
\rho^{\otimes n}
=
\bigoplus_{\lambda \in Y_d^{n}}
r_{\lambda}\,
\frac{\1^{(\mathcal{P}_\lambda)}}{m_\lambda}
\otimes
\rho_{\lambda}^{(\mathcal{Q}_\lambda)},
\qquad
L_c^{\times n}
=
\bigoplus_{\lambda \in Y_d^{n}}
\1^{(\mathcal{P}_\lambda)}
\otimes
L_{\lambda,c}^{(\mathcal{Q}_\lambda)}.
\end{equation}
Here, $\mathcal{P}_{\lambda}$ and $\mathcal{Q}_{\lambda}$ denote the Hilbert spaces carrying the irreducible representations of the symmetric group $S_n$ and the unitary group $U(d)$, respectively, with 
\begin{equation}\label{mlambda}
    m_{\lambda} := \dim  (\mathcal{P}_{\lambda}).
\end{equation}
Moreover, $Y_d^n$ denotes the set of Young diagrams with $n$ boxes and at most $d$ rows, or equivalently the set of partitions $\lambda=(\lambda_1,\ldots,\lambda_d)$ satisfying
$\lambda_1\ge\cdots\ge\lambda_d$ and $\sum_{i=1}^d\lambda_i=n$. We denote $\bar\lambda:=\frac{1}{n}\left(\lambda_1,\ldots,\lambda_d\right)$ as the normalized Young diagram

Using techniques from quantum spectrum estimation \cite{Keyl_Werner2001}, we show in Theorem~\ref{thm_8_important} that, for every $\epsilon_n>0$, there exists a projector $P_{\epsilon_n}$ supported on a subspace of $\mathcal{H}_d^{\otimes n}$ satisfying
\begin{equation}\label{dtt}
\left\|
\rho^{\otimes n}
-
\frac{P_{\epsilon_n}\rho^{\otimes n}P_{\epsilon_n}}
{\Tr\!\left(P_{\epsilon_n}\rho^{\otimes n}P_{\epsilon_n}\right)}
\right\|_1
\lesssim
\poly(n)
\exp\!\left[
-\frac{n}{4}
\left(
\frac{\epsilon_n}{\log(1/\epsilon_n)}
\right)^2
\right],
\end{equation}

where
\begin{equation}\label{employment_eq}
\frac{P_{\epsilon_n}\rho^{\otimes n}P_{\epsilon_n}}
{\Tr\!\left(P_{\epsilon_n}\rho^{\otimes n}P_{\epsilon_n}\right)}
=
\bigoplus_{|H(\bar\lambda)-s|\le\epsilon_n}
\tilde r_\lambda
\frac{\1^{(\mathcal P'_\lambda)}}
{\Tr\!\left(\1^{(\mathcal P'_\lambda)}\right)}
\otimes
\tilde\rho_\lambda^{(\mathcal Q'_\lambda)},
\end{equation}
with 
\begin{equation}
    \tilde r_{\lambda}=\frac{r_{\lambda}}{\sum_{\lambda:|H(\bar\lambda)-s|\leq \epsilon_n}r_{\lambda}},
\end{equation}
and $\mathcal Q'_\lambda\subseteq\mathcal Q_\lambda$. From Eq. \eqref{mlambda}, let us define
\begin{equation}
\label{dimm_lambda_defn}
D_{\min}^{(\epsilon_n)}
:=
\min_{\lambda:\, |H(\bar{\lambda})-S(\rho)|\leq \epsilon_n}
m_{\lambda}.
\end{equation}

The subspace $\E$ is isomorphic to a subspace contained in $\mathcal{P}_{\lambda}$ with 
\begin{equation}
   \dim (\E) =\floor{D_{\min}^{(\epsilon_n)}e^{-n\epsilon_n}}.
\end{equation}
Note that $\mathcal{E}$ is chosen in such a way that $\dim (\mathcal{E})$ is independent of $\lambda$. Now, consider $\mathcal{P}'_\lambda\subseteq\mathcal{P}_\lambda$ to be the largest subspace satisfying
\begin{equation}
m'_\lambda:=
\dim (\mathcal{P}'_\lambda)
=
\dim (\mathcal{E})\,f_\lambda=\floor{D_{\min}^{(\epsilon_n)}e^{-n\epsilon_n}}f_\lambda,
\end{equation}
where $f_\lambda$ is a positive integer. Consequently, for every $\lambda$ satisfying
$|H(\bar\lambda)-s|\le\epsilon_n$,
we can take any unitary 
\begin{equation}
    U_\lambda:\mathcal P'_\lambda
\rightarrow
\mathcal E\otimes\mathcal F_\lambda,
\end{equation}
where $\mathcal F_\lambda$ is a vector space with $\dim (\mathcal F_\lambda)=f_\lambda$. Defining the block-diagonal unitary acting on $\mathrm{supp}(P_{\epsilon_n})$ as
\begin{equation}
U=
\bigoplus_{|H(\bar\lambda)-s|\le\epsilon_n}
U_\lambda
\otimes
\mathbb I^{(\mathcal Q'_\lambda)},
\end{equation}
we obtain using Eq. \eqref{employment_eq}
\begin{equation}
U
\left(
\frac{P_{\epsilon_n}\rho^{\otimes n}P_{\epsilon_n}}
{\Tr(P_{\epsilon_n}\rho^{\otimes n}P_{\epsilon_n})}
\right)
U^\dagger
=
\frac{\mathbb I^{(\mathcal E)}}{\dim (\mathcal E)}
\otimes
\eta_\rho^{(\mathcal K)},
\end{equation}
where
\begin{equation}
\eta_\rho^{(\mathcal K)}
=
\bigoplus_{|H(\bar\lambda)-s|\le\epsilon_n}
\tilde r_\lambda
\frac{\mathbb I^{(\mathcal F_\lambda)}}
{\dim (\mathcal F_\lambda)}
\otimes
\tilde\rho_\lambda^{(\mathcal Q'_\lambda)}
\end{equation}
is a state supported on the Hilbert space
\begin{equation}
\mathcal K
=
\bigoplus_{|H(\bar\lambda)-s|\le\epsilon_n}
\mathcal F_\lambda
\otimes
\mathcal Q'_\lambda.
\end{equation}

Using spectrum-estimation techniques, we show (Theorems~\ref{charge_on_anc} and \ref{dimEbd}) that the dimensions of $\mathcal E$ and $\mathcal K$ satisfy the bounds in Eq.~\eqref{dimchargebd}. We then define the conserved charges on the ancillary Hilbert space $\mathcal K$ as
\begin{equation}\label{conanc}
L_c^{(\mathcal K)}
=
\bigoplus_{|H(\bar\lambda)-S(\rho)|\le\epsilon_n}
\mathbb I^{(\mathcal F_\lambda)}
\otimes
\left(
\tilde L_{\lambda,c}^{(\mathcal Q'_\lambda)}
-
nl_c\,\mathbb I^{(\mathcal Q'_\lambda)}
\right),
\end{equation}
where
\begin{equation}
\tilde L_{\lambda,c}^{(\mathcal Q'_\lambda)}
=
\mathbb I^{(\mathcal Q'_\lambda)}
L_{\lambda,c}^{(\mathcal Q_\lambda)},
\end{equation}
with $\mathbb I^{(\mathcal Q'_\lambda)}$ denoting the projector onto the subspace $\mathcal Q'_\lambda\subseteq\mathcal Q_\lambda$. It then follows (Theorem~\ref{U_lambda_defn_thm9232}) that
\begin{equation}
L_c^{(\mathcal E)}
\otimes
\mathbb I^{(\mathcal K)}
+
\mathbb I^{(\mathcal E)}
\otimes
L_c^{(\mathcal K)}
=
U(P_{\epsilon_n}L_c^{\times n})U^\dagger.
\end{equation}
Extending $U$ to the full Hilbert space via the direct-sum structure in Eq.~\eqref{dsumstru} establishes the charge-conservation relation \eqref{conscc}, while unitary invariance of the trace norm yields Eq.~\eqref{dt} from Eq.~\eqref{dtt}. Finally, Theorem~\ref{charge_on_anc} bounds the operator norm of $L_c^{(\mathcal K})$. The complete proof is given in Appendix~\ref{app1known}; see Subsection~\ref{cc1_calc} for details.


\subsection{Proof of Theorem \ref{corr1}}
\label{proof_of_corr}
With this choice, the dimension of the ancillary system satisfies
\begin{equation}
    \dim (\mathcal{H}_{\mathrm{anc}})
    \leq
    \poly(n)\exp\!\left(3n\epsilon_n\right)
    =
    \poly(n)\exp\!\left(3\sqrt{n}(\log n)^2\right).
\end{equation}
Equivalently, $\log\left(\dim (\mathcal{H}_{\mathrm{anc}})\right)
    =
    O\left(\sqrt{n}(\log n)^2\right)
    =
    o(n)$, showing that the ancilla size is sublinear in the number of copies. Moreover, the charge associated with the ancillary system is bounded as
\begin{equation}
    \|L_{c,\mathrm{anc}}\|
    \leq
    n\epsilon_n
    =
    \sqrt{n}(\log n)^2
    =
    o(n).
\end{equation}

The transformation error is bounded by
\begin{equation}
    \poly(n)
    \exp\!\left[
    -\frac{n}{4}
    \left(
    \frac{\epsilon_n}{\log(1/\epsilon_n)}
    \right)^2
    \right].
\end{equation}
For the above choice of $\epsilon_n$, we have $    \log(1/\epsilon_n)=\Theta(\log n)$. Therefore
\begin{equation}
    \frac{n}{4}
    \left(
    \frac{\epsilon_n}{\log(1/\epsilon_n)}
    \right)^2
    =
    \Theta((\log n)^2).
\end{equation}
Here, $f(n)=\Theta(g(n))$ means that there exist constants $c_1,c_2>0$ such that $ c_1g(n)\leq f(n)\leq c_2g(n)$ for sufficiently large $n$. Consequently,
\begin{equation}
    \poly(n)
    \exp\!\left[
    -\frac{n}{4}
    \left(
    \frac{\epsilon_n}{\log(1/\epsilon_n)}
    \right)^2
    \right]
    =
    \exp[-\Theta((\log n)^2)],
\end{equation}
which vanishes faster than any inverse polynomial in the limit $n\rightarrow\infty$.

\vfill\pagebreak

\appendix
\onecolumngrid
\section{Notation}
We introduce the notation that will be used throughout this manuscript in Table \ref{tab:notation}.
\begin{table*}[htbp!]
\centering
\caption{Summary of notation used throughout the manuscript.}
\label{tab:notation}

\renewcommand{\arraystretch}{1.4}
\setlength{\tabcolsep}{12pt}

\begin{tabular}{p{0.25\textwidth} p{0.65\textwidth}}
\hline
\textbf{Notation} & \textbf{Definition} \\
\hline

$\1^{(\mathcal{V})}$ 
& Identity operator on $\mathcal{V}$. More generally, the projector supported on $\mathcal{W}\subseteq\mathcal{V}$ is denoted by $\1^{(\mathcal{W})}$. This notation is justified by the fact that the projector onto a subspace acts as the identity operator on that subspace. \\[2mm]

$\dim (\mathcal{V})$ 
& Dimension of the vector space $\mathcal{V}$. \\[2mm]

$S_n$ 
& Symmetric group of $n$ objects. \\[2mm]

$U(d)$ 
& Unitary group in dimension $d$. \\[2mm]

$A^{\times n}$ 
& Total operator :
$A^{\times n}=A\otimes\1\otimes\ldots\otimes\1+\1\otimes\1\otimes\ldots\otimes A+\ldots+\1\otimes\1\otimes\ldots\otimes A$. \\[2mm]

$S(\rho)$ 
& Von Neumann entropy:
$S(\rho)=-\Tr(\rho\ln\rho)$. \\[2mm]

$\mathrm{supp}(A)$ 
& Support of the operator $A$. \\[2mm]

$\mathcal{V}\cong\mathcal{W}$ 
& $\mathcal{V}$ is isomorphic to $\mathcal{W}$. \\[2mm]

$\|A\|_1$ 
& Trace norm:
$\|A\|_1:=\Tr(\sqrt{A^\dagger A})$. \\[2mm]
l
$\|A\|$ 
& Operator norm:
$\|A\|:=\sup_{\|x\|=1}\|Ax\|
=\sqrt{\lambda_{\max}(A^\dagger A)}$. \\[2mm]

$f(n)=\mathcal{O}(g(n))$ 
& There exist constants $C>0$ and $n_0$ such that
$f(n)\leq Cg(n)$ for all $n\geq n_0$. \\[2mm]

$f(n)=\Theta(g(n))$ 
& There exist constants $C_1, C_2>0$ and $n_0$ such that
$C_1 g(n)\le f(n)\le C_2 g(n)$ for all $n\geq n_0$. \\[2mm]

$f(n)=o(g(n))$ 
& $\displaystyle\lim_{n\rightarrow\infty}\frac{f(n)}{g(n)}=0$. \\[2mm]

$f(n)\sim h(n)$ 
& $f(n)=h(n)\pm o(h(n))$. \\[2mm]

$\poly(n)$ 
& Polynomial function:
$\poly(n)=\mathcal{O}(n^k)$ for some constant $k>0$. \\

\hline
\end{tabular}

\end{table*}

\section{State transformation from a known initial state}\label{app1known}
\subsection{Exploring the permutation invariance of the independent and identically prepared state and the total charge}
In this section, we explore the permutation invariance of the independent and identically prepared state \( \rho^{\otimes n} \) and the total charge \( L_c^{\times n} \) acting on the \( n \)-fold tensor product space using tools from Schur--Weyl duality, a well-known framework from representation theory. This structure will be used in the proof of the main theorem. We proceed by recalling the definition of permutation invariance.

\begin{defn}[Permutation invariant operator]
Let \(V_\pi\) denote the unitary representation of a permutation \(\pi \in S_n\), defined on the computational basis as
\begin{equation}
    V_\pi\left( \ket{i_1} \otimes \cdots \otimes \ket{i_n} \right)
    =
    \ket{i_{\pi^{-1}(1)}} \otimes \cdots \otimes \ket{i_{\pi^{-1}(n)}}.
\end{equation}
An operator \(X\) on \(\mathcal{H}_d^{\otimes n}\) is said to be \emph{permutation invariant} if it commutes with all $V_\pi$, in other words it satisfies
\begin{equation}
\forall\, \pi \in S_n \quad V_\pi X V_\pi^\dagger = X.
\end{equation}
\end{defn}
It is straightforward to observe that both \( \rho^{\otimes n} \) and \( L_c^{\times n} \) are permutation invariant.  We now state Schur--Weyl duality, which captures the interplay between representations of the symmetric group \( S_n \) and the unitary group \( U(d) \) on the tensor product space \( \mathcal{H}_d^{\otimes n} \).
\begin{lem}[Schur--Weyl duality \cite{goodman_wallach_representation_theory}]
\label{SWDuality}
Let \(\mathcal{H}_d\) be a \(d\)-dimensional Hilbert space. Then
\begin{equation}\label{Schur_Weyl_Hilbert_space}
    \mathcal{H}_d^{\otimes n}
    =
    \bigoplus_{\lambda \in Y_d^n}
    \mathcal{P}_{\lambda}\otimes\mathcal{Q}_{\lambda},
\end{equation}
where \(\mathcal{P}_{\lambda}\) and \(\mathcal{Q}_{\lambda}\) carry irreducible representations of \(S_n\) and \(U(d)\), respectively. Here \(Y_d^n\) denotes the set of Young diagrams with \(n\) boxes and at most \(d\) rows, or equivalently partitions \(\lambda=(\lambda_1,\ldots,\lambda_d)\) satisfying
\(\lambda_1 \ge \cdots \ge \lambda_d\) and \(\sum_{i=1}^d \lambda_i = n\). Any permutation-invariant operator \(X\) on \(\mathcal{H}_d^{\otimes n}\) admits a block-diagonal decomposition
\begin{equation}\label{Schur_Weyl_Permutation_symmetric}
    X
    =
    \bigoplus_{\lambda \in Y_d^n}
    \1^{(\mathcal{P}_{\lambda})}
    \otimes
    X^{(\mathcal{Q}_{\lambda})}.
\end{equation}
\end{lem}
 We will denote the
\begin{equation}
\label{dimplql}
    \dim(\mathcal{P}_{\lambda}) := m_{\lambda}
    \quad\text{and}\quad \dim(\mathcal{Q}_{\lambda}) := n_{\lambda}.
\end{equation}
It is well known that $|Y_d^n|=\poly(n)$ and $ \dim (\mathcal{Q}_{\lambda})=\poly(n)$~\cite{harrow2005applications}. Exploiting the permutation invariance of both the state \(\rho^{\otimes n}\) and the total charges \(L_c^{\times n}\) for \(c \in \{1,\ldots,t\}\), using Schur--Weyl duality, we can write
\begin{equation}\label{Tensorstate_totalcharge}
\rho^{\otimes n}
=
\bigoplus_{\lambda \in Y_d^{n}}
r_{\lambda}\,
\frac{\1^{(\mathcal{P}_\lambda)}}{m_\lambda}
\otimes
\rho_{\lambda}^{(\mathcal{Q}_\lambda)},
\qquad
L_c^{\times n}
=
\bigoplus_{\lambda \in Y_d^{n}}
\1^{(\mathcal{P}_\lambda)}
\otimes
L_{\lambda,c}^{(\mathcal{Q}_\lambda)}.
\end{equation}
If the charges commute pairwise, i.e.,
\begin{equation}
[L_c, L_{c'}]=0 \qquad \forall\, c,c' \in \{1,\ldots,t\},
\end{equation}
then there exists a common eigenbasis \(\{|x\rangle\}_{x=1}^d\) that simultaneously diagonalizes all the charges. In this basis, each charge can be written as
\begin{equation}
    L_c = \sum_{x=1}^{d} L_c(x)\, |x\rangle\langle x|
    \qquad \text{for } c \in \{1,\ldots,t\},
\end{equation}
where \(L_c(x)\) denotes the eigenvalue of \(L_c\) corresponding to the basis vector \(|x\rangle\). Using this eigendecomposition, the total charge can be expressed as
\begin{equation}\label{basis_states}
    L_{c}^{\times n}
    =
    \sum_{x_1, \ldots, x_n}
    \big(L_c(x_1) + \cdots + L_c(x_n)\big)
    |x_1, \ldots, x_n \rangle \langle x_1, \ldots, x_n|.
\end{equation}
Note that the spectrum of \(L_c^{\times n}\) is degenerate for every \(n>1\), since any two product basis states related by a permutation have the same eigenvalue. Since \(L_c^{\times n}\) is permutation invariant, it follows that the projector onto each degenerate eigenspace is also permutation invariant. For our purposes, we will be interested in permutation-invariant projectors whose corresponding eigenvalues are close to the expectation value \(\mathrm{Tr}(\rho L_c) = l_c\) for all $c\in\{1,\ldots,t\}$.
\begin{prop}
The projector \(\Pi_{\epsilon_n}\) on the subspace
\begin{equation}\label{Piepsilonn}
    \operatorname{span}\left\{
    |x_1,\ldots,x_n\rangle \quad\text{such that}\quad
    \forall c\in\{1,\ldots,t\},\;
    \left|\frac{1}{n}\sum_{i=1}^n L_c(x_i) - l_c\right| \le \epsilon_n
    \right\}
\end{equation}
is permutation invariant, i.e.,
\[
V_{\pi}\Pi_{\epsilon_n}V_{\pi}^{\dagger}=\Pi_{\epsilon_n}
\qquad \forall\, \pi\in S_n,
\]
where \(l_c = \mathrm{Tr}(\rho L_c)\). 
\end{prop}
\begin{proof}
For \(\epsilon_n > 0\), define the projectors \(\Pi_{c,\epsilon_n}\) from the eigen decomposition of \(L_c^{\times n}\) given in Eq. \eqref{basis_states}
\begin{equation}
    \Pi_{c,\epsilon_n}
    :=
    \sum_{x_1,\ldots,x_n}
    |x_1,\ldots,x_n\rangle\langle x_1,\ldots,x_n|
    \;\; \text{such that} \;\;
    \left|\frac{1}{n}\sum_{i=1}^n L_c(x_i) - l_c\right| \le \epsilon_n.
\end{equation}
Since all charges \(L_c\) commute, the projectors \(\Pi_{c,\epsilon_n}\) commute for all \(c,c' \in \{1,\ldots,t\}\), i.e.,
\begin{equation}
    [\Pi_{c,\epsilon_n},\Pi_{c',\epsilon_n}] = 0.
\end{equation}
We therefore construct the projector by taking product of all $\{\Pi_{c,\epsilon_n}\}_{c=1}^t$ for all \(c \in \{1,\ldots,t\}\) as follows:
\begin{equation}\label{charge_projector}
    \prod_{c=1}^{t} \Pi_{c,\epsilon_n},
\end{equation}
and by construction, this projector is supported on the subspace
\begin{equation}
    \operatorname{span}\left\{
    |x_1,\ldots,x_n\rangle \quad\text{such that}\quad
    \forall c\in\{1,\ldots,t\},\;
    \left|\frac{1}{n}\sum_{i=1}^n L_c(x_i) - l_c\right| \le \epsilon_n
    \right\}.
\end{equation}
Hence 
\begin{equation}\label{charge_projector2}
    \prod_{c=1}^{t} \Pi_{c,\epsilon_n}=\Pi_{\epsilon_n}.
\end{equation}
Moreover, each \(\Pi_{c,\epsilon_n}\) is permutation invariant since it is a spectral projector of the permutation-invariant operator \(L_c^{\times n}\). Hence \(\Pi_{\epsilon_n}\) is also permutation invariant since it is constructed as a product of the permutation-invariant projectors.
\end{proof}
As we have shown that $\Pi_{\epsilon_n}$ is permutation invariant, Schur--Weyl duality (Lemma~\ref{SWDuality}) implies that
\begin{equation}\label{energy_typ_proj}
    \Pi_{\epsilon_n}
    =
    \bigoplus_{\lambda \in Y_d^n}
    \1^{(\mathcal{P}_{\lambda})}
    \otimes
    \Pi^{(\mathcal{Q}_{\lambda})}
    =
    \bigoplus_{\lambda \in Y_d^n}
    \1^{(\mathcal{P}_{\lambda})}
    \otimes
    \1^{(\mathcal{Q}'_{\lambda})},
\end{equation}
where
\begin{equation}\label{Qprimelambda_defn}
\mathcal{Q}'_{\lambda}
:=
\operatorname{supp}\!\left(\Pi^{(\mathcal{Q}_{\lambda})}\right)
\subseteq
\mathcal{Q}_{\lambda}.
\end{equation}
The second equality follows from the fact that every projector acts as the identity on its support.

\subsection{Approximating the independent and identically prepared state via projection} 

We now turn to the consequences of the permutation invariance of the product state \(\rho^{\otimes n}\), which allows us to approximate it as a tensor product of a maximally mixed state, which does not depend explicitly on \(\rho\), and a non-trivial component that retains the dependence on \(\rho\). To proceed, we rely on a central result originally established in \cite{Keyl_Werner2001} in the context of spectrum estimation of an unknown quantum state. This result has since been further developed in quantum tomography \cite{WrightOdonnel,Haah} and in variable-length quantum source coding \cite{Hayashi2002}. For completeness, we present the relevant formulation of the result together with its proof.

\begin{lem}\label{Keyl_Werner}
   For any $\epsilon_n>0$, there exists $\eta_{\epsilon_n}>0$ such that, for every state $\rho$,
\begin{equation}\label{L1}
    \Tr\left(\tilde{\Pi}_{\epsilon_n}\rho^{\otimes n}\right)
    \geq
    1-\poly(n)e^{-n\eta_{\epsilon_n}},
\end{equation}
where
\begin{equation}\label{entropy_proj}
    \tilde{\Pi}_{\epsilon_n}
    :=
    \bigoplus_{\lambda:\,|H(\bar{\lambda})-S(\rho)|\le\epsilon_n}
    \1^{(\mathcal P_\lambda)}
    \otimes
    \1^{(\mathcal Q_\lambda)}\qquad\text{and}\qquad \bar{\lambda}=\frac{1}{n} (\lambda_1,\ldots,\lambda_d).
\end{equation}
Consequently,
\begin{equation}\label{GOL}
    \left\|
        \rho^{\otimes n}
        -
        \frac{\tilde{\Pi}_{\epsilon_n}
        \rho^{\otimes n}
        \tilde{\Pi}_{\epsilon_n}}{\Tr\left(\tilde{\Pi}_{\epsilon_n}
        \rho^{\otimes n}
        \tilde{\Pi}_{\epsilon_n}\right)}
    \right\|_{1}
    \le
    2\sqrt{\poly(n)e^{-n\eta_{\epsilon_n}}}.
\end{equation}
\end{lem}
\begin{proof}
We begin by recalling the Fannes--Audenaert inequality. It states that for any two quantum states $\rho$ and $\sigma$ on a $d$-dimensional Hilbert space satisfying
\begin{equation}
    \frac{1}{2}\|\rho-\sigma\|_{1}\leq\delta\leq1,
\end{equation}
their von Neumann entropies obey
\begin{equation}\label{Alicki_Fannes_Aud_ineq}
    |S(\rho)-S(\sigma)|
    \leq
    \delta\log d+H_2(\delta)
    =:
    \xi(\delta),
\end{equation}
where $H_2(\delta)$ denotes the binary entropy,
\begin{equation}
    H_2(\delta)
    =
    -\delta\log\delta
    -(1-\delta)\log(1-\delta).
\end{equation}
Applying the Fannes--Audenaert inequality from Eq. \eqref{Alicki_Fannes_Aud_ineq} to the pair of states $\rho$ and $\text{Diag}(\bar{\lambda})$, we obtain that, for every $\lambda\in Y_d^n$,
\begin{equation}
    \|\text{Diag}(\bar{\lambda})-\rho\|_1
    \le
    \xi^{-1}(\epsilon_n)
    \quad\Longrightarrow\quad
    |H(\bar{\lambda})-S(\rho)|
    \le
    \epsilon_n,
\end{equation}
where $\text{Diag}(\bar{\lambda})$ denotes a diagonal density with spectrum $\bar{\lambda} = \frac{1}{n}(\lambda_1,\ldots,\lambda_d)$. Consequently,
\begin{equation}\label{inclusion}
    \left\{
    \lambda\in Y_d^n:
    |H(\bar{\lambda})-S(\rho)|
    \le
    \epsilon_n
    \right\}\supseteq \left\{
    \lambda\in Y_d^n:
    \|\text{Diag}(\bar{\lambda})-\rho\|_1
    \le
    \xi^{-1}(\epsilon_n)
    \right\}.
\end{equation}
Using the inclusion from Eq. \eqref{inclusion} we can write
\begin{align}
    \Tr\left(\tilde{\Pi}_{\epsilon_n}\rho^{\otimes n}\right)
    =
    \Tr\left[
    \left(
    \bigoplus_{\lambda:\,
    |H(\bar{\lambda})-S(\rho)|\le\epsilon_n}
    \1^{(\mathcal P_\lambda)}
    \otimes
    \1^{(\mathcal Q_\lambda)}
    \right)
    \rho^{\otimes n}
    \right]
    &\overset{(1)}{\geq}
    \Tr\left[
    \left(
    \bigoplus_{\lambda:\,
    \|\text{Diag}(\bar{\lambda})-\rho\|_1
    \le
    \xi^{-1}(\epsilon_n)}
    \1^{(\mathcal P_\lambda)}
    \otimes
    \1^{(\mathcal Q_\lambda)}
    \right)
    \rho^{\otimes n}
    \right]
    \nonumber\\
    &\overset{(2)}{\geq}
    1-\poly(n)e^{-n\eta_{\epsilon_n}}\label{fi},
\end{align}
where
\(
\eta_{\epsilon_n}
=
\tfrac12\bigl(\xi^{-1}(\epsilon_n)\bigr)^2.
\)
Here, the inequality $(\overset{(1)}{\geq})$ follows from Eq.~\eqref{inclusion}. The inequality $(\overset{(2)}{\geq})$ in Eq.~\eqref{fi} is obtained using \cite[Eq.~(6.23)]{harrow2005applications}, which is itself derived from the result of Keyl and Werner \cite{Keyl_Werner2001}. Finally, using the Gentle Operator Lemma \ref{Gentle_op_lem}, we can obtain Eq.~\eqref{GOL}.
\end{proof}

\begin{lem}
  \label{amc}
    For any $\epsilon_n>0$, there exists $\alpha({\epsilon_n})>0$ such that for any states $\rho^{\otimes n}$,
    \begin{equation}\label{L2}
        \Tr\left(\Pi_{\epsilon_n}\rho^{\otimes n}\right) \geq 1-2\left(\sum_{c=1}^t\mathrm{exp}\left(-\frac{2n\epsilon_n^2}{(\Delta L_c)^2}\right)\right) \geq 1-2te^{-2n\alpha({\epsilon_n}) },
    \end{equation}
    where $\Pi_{\epsilon_n}$ given in Eq. \eqref{energy_typ_proj}, $t$ is the total number of commuting charges, $\Delta L_c$ is the difference between largest and the smallest eigenvalue of $L_c$. Consequently,
\begin{equation}\label{GOL2}
    \left\|
        \rho^{\otimes n}
        -
        \frac{\Pi_{\epsilon_n}
        \rho^{\otimes n}
        \Pi_{\epsilon_n}}{\Tr\left(\Pi_{\epsilon_n}
        \rho^{\otimes n}
        \Pi_{\epsilon_n}\right)}
    \right\|_{1}
    \le
    2\sqrt{2te^{-2n\alpha({\epsilon_n})}}.
\end{equation}
\end{lem}
\begin{proof}
We proceed by introducing
\begin{equation}\label{Defn_Delta_Lk_max}
    \Delta L_{\max}
    :=
    \max_{c\in\{1,\ldots,t\}}
    \Delta L_c.
\end{equation}
Using the union bound together with Hoeffding's inequality (Theorem~\ref{Hoeffding_ineq}), we obtain
\begin{align}
1-\Tr(\Pi_{\epsilon_n}\rho^{\otimes n})
&=
\Pr\!\left[
\bigcup_{c=1}^{t}
\left|\frac{1}{n}\sum_{i=1}^n L_c(x_i) - l_c\right| >\epsilon_n
\right]
\nonumber\\
&\overset{(1)}{\leq}
\sum_{c=1}^{t}
\Pr\!\left[
\left|\frac{1}{n}\sum_{i=1}^n L_c(x_i) - l_c\right| >\epsilon_n
\right]
\nonumber\\
&\overset{(2)}{\leq}
2\sum_{c=1}^{t}
\exp\!\left(
-\frac{2n\epsilon_n^2}{(\Delta L_c)^2}
\right)
\nonumber\\
&\overset{(3)}{\leq}
2t
\exp\!\left(
-\frac{2n\epsilon_n^2}{(\Delta L_{\max})^2}
\right).
\label{AMC_error}
\end{align}
Here, the first inequality is the union bound, the second follows from Hoeffding's inequality, and the third uses the definition of $\Delta L_{\max}$ in Eq.~\eqref{Defn_Delta_Lk_max}. Finally, defining
\begin{equation}
    \alpha({\epsilon_n})
    :=
    \frac{\epsilon_n^2}{(\Delta L_{\max})^2},
\end{equation}
Eq.~\eqref{AMC_error} can be rewritten as
\begin{equation}
    \Tr(\Pi_{\epsilon_n}\rho^{\otimes n})
    \ge
    1
    -
    2t\,e^{-2n\alpha({\epsilon_n})}.
\end{equation}
Finally, using the Gentle Operator Lemma \ref{Gentle_op_lem}, we can obtain Eq.~\eqref{GOL2}.
\end{proof}

\begin{thm}
\label{ref100}
   Let $\Pi_{\epsilon_n}$ and $\tilde{\Pi}_{\epsilon_n}$ be the projectors introduced in Eq.~\eqref{charge_projector2} and Eq.~\eqref{entropy_proj}, respectively. Then
\begin{equation}\label{ref200}
    \left\|
    \rho^{\otimes n}
    -
    \frac{
    \Pi_{\epsilon_n}\tilde{\Pi}_{\epsilon_n}
    \rho^{\otimes n}
    \tilde{\Pi}_{\epsilon_n}\Pi_{\epsilon_n}
    }{
    \Tr\left(
    \Pi_{\epsilon_n}\tilde{\Pi}_{\epsilon_n}
    \rho^{\otimes n}
    \tilde{\Pi}_{\epsilon_n}\Pi_{\epsilon_n}
    \right)
    }
    \right\|_1
    \leq
    2\sqrt{
    \poly(n)e^{-n\eta_{\epsilon_n}}
    +
    2t\,e^{-2n\alpha({\epsilon_n})}
    } ,
\end{equation}
and
\begin{equation}
    \Tr\left(
    \Pi_{\epsilon_n}\tilde{\Pi}_{\epsilon_n}
    \rho^{\otimes n}
    \tilde{\Pi}_{\epsilon_n}\Pi_{\epsilon_n}
    \right)
    \geq
    1-
    4\left(
    \poly(n)e^{-n\eta_{\epsilon_n}}
    +
    2t\,e^{-2n\alpha({\epsilon_n})}
    \right).
\end{equation}
\end{thm}
\begin{proof}
   Starting from the inequalities in Eqs.~\eqref{L1} and \eqref{L2}, established in Lemmas~\ref{Keyl_Werner} and \ref{amc}, respectively, the proof follows directly by applying the non-commutative union bound stated in Theorem~\ref{NCUB}.
\end{proof}

\begin{thm}
\label{thm_8_important}
For any \(\epsilon_n>0\), recall the definition [Eq.~\eqref{dimm_lambda_defn}]
\begin{equation}
D_{\min}^{(\epsilon_n)}
:=
\min_{\lambda:\, |H(\bar{\lambda})-S(\rho)|\leq \epsilon_n}
m_{\lambda},
\end{equation}
where \(m_{\lambda}=\dim(\mathcal{P}_{\lambda})\) is defined in Eq.~\eqref{dimplql}. Then, any projector of the form
\begin{equation}\label{Pen}
P_{\epsilon_n}
=
\bigoplus_{\lambda:\,|H(\bar\lambda)-S(\rho)|\leq\epsilon_n}
\1^{(\P'_{\lambda})}\otimes\1^{(\Q'_{\lambda})},
\end{equation}
where \(\mathcal Q'_{\lambda}\) defined in Eq. \eqref{Qprimelambda_defn} and  \(\mathcal P'_{\lambda}\) is the largest-dimensional subspace of $\mathcal P_{\lambda}$ whose dimension is an integer multiple of
\(\left\lfloor D_{\min}^{(\epsilon_n)}e^{-n\epsilon_n}\right\rfloor\),
satisfies the following
\begin{enumerate}
\item Dimension bound:
\begin{equation}\label{Pen2}
\operatorname{supp}(P_{\epsilon_n})
\subseteq
\operatorname{supp}\left(\Pi_{\epsilon_n}\tilde{\Pi}_{\epsilon_n}\right),
\quad\text{and}\quad
\operatorname{rank}(P_{\epsilon_n})
=
\operatorname{poly}(n)
\left\lfloor
D_{\min}^{(\epsilon_n)}e^{-n\epsilon_n}
\right\rfloor
\left\lfloor
\frac{e^{n\epsilon_n}m_{\lambda}}
{D_{\min}^{(\epsilon_n)}}
\right\rfloor;
\end{equation}
\item Approximation:
\begin{equation}\label{bdtrace}
\left\|
\frac{
\Pi_{\epsilon_n}\tilde{\Pi}_{\epsilon_n}
\rho^{\otimes n}
\tilde{\Pi}_{\epsilon_n}\Pi_{\epsilon_n}
}{
\Tr\left(
\Pi_{\epsilon_n}\tilde{\Pi}_{\epsilon_n}
\rho^{\otimes n}
\tilde{\Pi}_{\epsilon_n}\Pi_{\epsilon_n}
\right)
}
-
\frac{
P_{\epsilon_n}
\rho^{\otimes n}
P_{\epsilon_n}
}{
\Tr\left(
P_{\epsilon_n}
\rho^{\otimes n}
P_{\epsilon_n}
\right)
}
\right\|_1
\leq
2e^{-n\epsilon_n/2}.
\end{equation}
\end{enumerate}
\end{thm}

\begin{proof}
We begin by expressing the projector \(\Pi_{\epsilon_n}\tilde{\Pi}_{\epsilon_n}\) using the Schur--Weyl decomposition from Eqs.~\eqref{energy_typ_proj} and \eqref{entropy_proj} as
\begin{equation}\label{supp_eqn}
\Pi_{\epsilon_n}\tilde{\Pi}_{\epsilon_n}
=
\bigoplus_{\lambda:\,|H(\bar\lambda)-S(\rho)|\leq \epsilon_n}
\1^{(\P_{\lambda})}\otimes \1^{(\Q'_{\lambda})}.
\end{equation}
Consequently, using Eq.~\eqref{Pen}, we obtain
\begin{equation}\label{supp_eqn2}
\operatorname{supp}\left(\Pi_{\epsilon_n}\tilde{\Pi}_{\epsilon_n}\right)
=
\bigoplus_{\lambda:\,|H(\bar\lambda)-S(\rho)|\leq \epsilon_n}
\P_{\lambda}\otimes\Q'_{\lambda}.
\end{equation}
For each \(\lambda\) satisfying
\(\left|H(\bar{\lambda})-S(\rho)\right|\leq\epsilon_n\), since
\(\P'_{\lambda}\subseteq\P_{\lambda}\), Eq.~\eqref{Pen} directly implies
\begin{equation}
    \operatorname{supp}(P_{\epsilon_n})
    \subseteq
    \operatorname{supp}\left(\Pi_{\epsilon_n}\tilde{\Pi}_{\epsilon_n}\right).
\end{equation}
As \(\mathcal{P}'_{\lambda}\) is the largest-dimensional subspace of \(\mathcal{P}_{\lambda}\) whose dimension is an integer multiple of \(\left\lfloor D_{\min}^{(\epsilon_n)}e^{-n\epsilon_n}\right\rfloor\), the dimension \(\dim(\mathcal{P}'_{\lambda})\) can be decomposed into a \(\lambda\)-independent factor and a \(\lambda\)-dependent factor as follows:
\begin{equation}\label{dimm_prime_lambda}
\dim (\mathcal{P}'_{\lambda})=\Tr\left(\1^{\left(\P'_{\lambda}\right)}\right)
=
\left\lfloor
D_{\min}^{(\epsilon_n)}e^{-n\epsilon_n}
\right\rfloor
\left\lfloor
\frac{e^{n\epsilon_n}m_{\lambda}}
{D_{\min}^{(\epsilon_n)}}
\right\rfloor
=:m'_{\lambda}.
\end{equation}
We also define the orthogonal complement of \(\mathcal{P}'_{\lambda}\) in
\(\mathcal{P}_{\lambda}\) as \(\mathcal{P}_{\lambda}^{'\perp}\) for each
\(\lambda\) satisfying \(|H(\bar\lambda)-S(\rho)|\leq\epsilon_n\). Hence,
\begin{equation}\label{ocdim}
    \P_{\lambda}
    =
    \mathcal{P}'_{\lambda}\oplus\mathcal{P}_{\lambda}^{'\perp},
    \quad\text{which implies}\quad
    m_{\lambda}
    =
    m'_{\lambda}
    +
    \dim\left(\mathcal{P}_{\lambda}^{'\perp}\right).
\end{equation}
As $\dim (\Q'_{\lambda})=\mathrm{poly}(n)$ and the cardinality of the set $\left\{\lambda:|H(\bar\lambda)-S(\rho)|\leq \epsilon_n\right\}$ is $\mathrm{poly}(n)$, too, we can obtain the $\mathrm{rank}\left(P_{\epsilon_n}\right)$ by taking the trace of $P_{\epsilon_n}$ given in Eq. \eqref{Pen},
\begin{equation}
    \mathrm{rank}\left(P_{\epsilon_n}\right) = \Tr\left(P_{\epsilon_n}\right)=\sum_{\lambda:|H(\bar\lambda)-S(\rho)|\leq\epsilon_n}\dim (\P'_{\lambda})\dim (\Q'_{\lambda})=\sum_{\lambda:|H(\bar\lambda)-S(\rho)|\leq\epsilon_n}m'_{\lambda}\dim (\Q'_{\lambda})=\operatorname{poly}(n)
\left\lfloor
D_{\min}^{(\epsilon_n)}e^{-n\epsilon_n}
\right\rfloor
\left\lfloor
\frac{e^{n\epsilon_n}D_{\max}^{(\epsilon_n)}}
{D_{\min}^{(\epsilon_n)}}
\right\rfloor,
\end{equation}
where in the last equality we use Eq. \eqref{dimm_prime_lambda}. Finally, to prove the inequality in Eq. \eqref{bdtrace}, we proceed by using Eqs. \eqref{supp_eqn} and  \eqref{Tensorstate_totalcharge} to calculate
\begin{equation}
\label{calculate}
    \frac{\Pi_{\epsilon_n}\tilde\Pi_{\epsilon_n}\rho^{\otimes n}\tilde\Pi_{\epsilon_n}\Pi_{\epsilon_n}}{\Tr\left(\Pi_{\epsilon_n}\tilde\Pi_{\epsilon_n}\rho^{\otimes n}\tilde\Pi_{\epsilon_n}\Pi_{\epsilon_n}\right)} =  \bigoplus_{\lambda:\,|H(\bar\lambda)-S(\rho)|\leq \epsilon_n} 
    \tilde{r}_{\lambda} \frac{\1^{(\mathcal{P}_{\lambda})}}{m_{\lambda}}\otimes 
    \tilde{\rho}_{\lambda}^{(\Q'_{\lambda})},
\end{equation}
where 
\begin{equation}\label{normalization}
    \frac{r_{\lambda}}{\left(\sum_{\lambda:\,|H(\bar\lambda)-S(\rho)|\leq \epsilon_n} r_{\lambda}\right)}:=\tilde{r}_{\lambda}\quad\mathrm{and}\quad\tilde{\rho}_{\lambda}^{(\Q'_{\lambda})} := \frac{\1^{\left(\mathcal{Q}^{'}_{\lambda}\right)}\rho_{\lambda}^{(\mathcal{Q}_{\lambda})}\1^{\left(\mathcal{Q}^{'}_{\lambda}\right)}}{\Tr\left(\1^{\left(\mathcal{Q}^{'}_{\lambda}\right)}\rho_{\lambda}^{(\mathcal{Q}_{\lambda})}\1^{\left(\mathcal{Q}^{'}_{\lambda}\right)}\right)}.
\end{equation}
Using Eqs. \eqref{Pen} and \eqref{Tensorstate_totalcharge}, we obtain the following:
\begin{align}
    \Tr\left(P_{\epsilon_n}\frac{\Pi_{\epsilon_n}\tilde\Pi_{\epsilon_n}\rho^{\otimes n}\tilde\Pi_{\epsilon_n}\Pi_{\epsilon_n}}{\Tr\left(\Pi_{\epsilon_n}\tilde\Pi_{\epsilon_n}\rho^{\otimes n}\tilde\Pi_{\epsilon_n}\Pi_{\epsilon_n}\right)}\right) &= 
    \Tr\Bigg(\Bigg[\bigoplus_{\lambda:\,|H(\bar\lambda)-S(\rho)|\leq \epsilon_n} 
    \1^{\left(\mathcal{P}'_{\lambda}\right)}\otimes \1^{\left(\mathcal{Q}'_{\lambda}\right)}\Bigg]
    \Bigg[\bigoplus_{\lambda:\,|H(\bar\lambda)-S(\rho)|\leq \epsilon_n} 
    \tilde{r}_{\lambda} \frac{\1^{(\mathcal{P}_{\lambda})}}{m_{\lambda}}\otimes 
    \tilde{\rho}_{\lambda}^{(\Q'_{\lambda})}\Bigg]\Bigg)\nonumber\\
    &\overset{(1)}{=} 
    \Tr\Bigg(\bigoplus_{\lambda:\,|H(\bar\lambda)-S(\rho)|\leq \epsilon_n} 
    \tilde{r}_{\lambda} \frac{\1^{(\mathcal{P}'_{\lambda})}}{m_{\lambda}}\otimes \tilde{\rho}_{\lambda}^{(\Q'_{\lambda})}\Bigg)\overset{(2)}{=} \sum_{\lambda:\,|H(\bar\lambda)-S(\rho)|\leq \epsilon_n}
    \tilde{r}_{\lambda}
    \frac{m'_{\lambda}}{m_{\lambda}}\nonumber\\
    &\overset{(3)}{=} 
    \sum_{\lambda:\,|H(\bar\lambda)-S(\rho)|\leq \epsilon_n}
    \tilde{r}_{\lambda}
    \left(1-\frac{\dim \left(\mathcal{P}_{\lambda}^{'\perp}\right)}{m_{\lambda}}\right)\overset{(4)}{=} 1-\sum_{\lambda:\,|H(\bar\lambda)-S(\rho)|\leq \epsilon_n}
    \tilde{r}_{\lambda}\times \frac{\dim \left(\mathcal{P}_{\lambda}^{'\perp}\right)}{m_{\lambda}}\nonumber\\
    &\overset{(5)}{\geq} 1-e^{-n\epsilon_n }\left(\sum_{\lambda:\,|H(\bar\lambda)-S(\rho)|\leq \epsilon_n}\tilde{r}_{\lambda}\right)
    \overset{(6)}{=} 1-e^{-n\epsilon_n },\label{last_inequality_term2}
\end{align}
where the first equality follows from the inclusion
\(\mathcal{P}'_{\lambda}\subseteq\mathcal{P}_{\lambda}\), the second equality follows from
\(\operatorname{Tr}\!\left(\1^{(\P'_{\lambda})}\right)=m'_{\lambda}\) as given in Eq.~\eqref{dimm_prime_lambda}, together with
\(\operatorname{Tr}\!\left(\tilde{\rho}_{\lambda}^{(\Q'_{\lambda})}\right)=1\) which follows from Eq.~\eqref{normalization}, the third equality follows from Eq.~\eqref{ocdim}, and the fourth equality follows from the normalization condition in Eq.~\eqref{normalization}. To write the inequality in Eq.~\eqref{last_inequality_term2}, we use Proposition \ref{prop_for_thm}. As $\operatorname{supp}(P_{\epsilon_n})
\subseteq
\operatorname{supp}\left(\Pi_{\epsilon_n}\tilde{\Pi}_{\epsilon_n}\right)$ we have $P_{\epsilon_n}\Pi_{\epsilon_n}\tilde{\Pi}_{\epsilon_n}=P_{\epsilon_n}$ which allows us to write
\begin{equation}\label{up}
     \frac{
P_{\epsilon_n}\Pi_{\epsilon_n}\tilde{\Pi}_{\epsilon_n}
\rho^{\otimes n}
\tilde{\Pi}_{\epsilon_n}\Pi_{\epsilon_n}P_{\epsilon_n}
}{
\Tr\left(
P_{\epsilon_n}\Pi_{\epsilon_n}\tilde{\Pi}_{\epsilon_n}
\rho^{\otimes n}
\tilde{\Pi}_{\epsilon_n}\Pi_{\epsilon_n}P_{\epsilon_n}
\right)
}=\frac{
P_{\epsilon_n}
\rho^{\otimes n}P_{\epsilon_n}
}{
\Tr\left(
P_{\epsilon_n}
\rho^{\otimes n}P_{\epsilon_n}\right)
}.
\end{equation}
Now using Eq. \eqref{last_inequality_term2}, one can employ the Gentle Operator Lemma \ref{Gentle_op_lem} to obtain 
\begin{equation}
    \left\|
\frac{
\Pi_{\epsilon_n}\tilde{\Pi}_{\epsilon_n}
\rho^{\otimes n}
\tilde{\Pi}_{\epsilon_n}\Pi_{\epsilon_n}
}{
\Tr\left(
\Pi_{\epsilon_n}\tilde{\Pi}_{\epsilon_n}
\rho^{\otimes n}
\tilde{\Pi}_{\epsilon_n}\Pi_{\epsilon_n}
\right)
}
-\frac{
P_{\epsilon_n}\Pi_{\epsilon_n}\tilde{\Pi}_{\epsilon_n}
\rho^{\otimes n}
\tilde{\Pi}_{\epsilon_n}\Pi_{\epsilon_n}P_{\epsilon_n}
}{
\Tr\left(
P_{\epsilon_n}\Pi_{\epsilon_n}\tilde{\Pi}_{\epsilon_n}
\rho^{\otimes n}
\tilde{\Pi}_{\epsilon_n}\Pi_{\epsilon_n}P_{\epsilon_n}
\right)
}\right\|_1=\left\|
\frac{
\Pi_{\epsilon_n}\tilde{\Pi}_{\epsilon_n}
\rho^{\otimes n}
\tilde{\Pi}_{\epsilon_n}\Pi_{\epsilon_n}
}{
\Tr\left(
\Pi_{\epsilon_n}\tilde{\Pi}_{\epsilon_n}
\rho^{\otimes n}
\tilde{\Pi}_{\epsilon_n}\Pi_{\epsilon_n}
\right)
}
-
\frac{
P_{\epsilon_n}
\rho^{\otimes n}
P_{\epsilon_n}
}{
\Tr\left(
P_{\epsilon_n}
\rho^{\otimes n}
P_{\epsilon_n}
\right)
}
\right\|_1
\leq
2e^{-n\epsilon_n/2},
\end{equation}
where the equality follows from Eq. \eqref{up}. This completes the proof.
\end{proof}
\begin{prop}\label{prop_for_thm}
For any $\epsilon_n>0$
    \begin{equation}\label{trace_norm_remainder_small}
    \forall\, \lambda:\, |H(\bar{\lambda}) - S(\rho)| \leq \epsilon_n, \quad
    \frac{\dim \left(\mathcal{P}_{\lambda}^{'\perp}\right)}{m_{\lambda}}
    <\frac{\left\lfloor
D_{\min}^{(\epsilon_n)}e^{-n\epsilon_n}
\right\rfloor}{D^{(\epsilon_n)}_{\min}}
    \leq e^{-n\epsilon_n }.
\end{equation}
\end{prop}
\begin{proof}
Recall \(\mathcal{P}'_{\lambda}\) is the largest-dimensional subspace of
\(\mathcal{P}_{\lambda}\) whose dimension is an integer multiple of
\(\left\lfloor D_{\min}^{(\epsilon_n)}e^{-n\epsilon_n}\right\rfloor\). We will do the proof by contradiction. If
\[
\dim\left(\mathcal{P}_{\lambda}^{'\perp}\right)
\ge
\left\lfloor D_{\min}^{(\epsilon_n)}e^{-n\epsilon_n}\right\rfloor,
\]
then one could enlarge \(\mathcal{P}'_{\lambda}\) by a subspace of
\(\mathcal{P}_{\lambda}^{'\perp}\) of dimension
\(\left\lfloor D_{\min}^{(\epsilon_n)}e^{-n\epsilon_n}\right\rfloor\), contradicting the maximality of \(\mathcal{P}'_{\lambda}\). Therefore,
\begin{equation}\label{remainderdimbound}
\dim\left(\mathcal{P}_{\lambda}^{'\perp}\right)
<
\left\lfloor D_{\min}^{(\epsilon_n)}e^{-n\epsilon_n}\right\rfloor
\le
D_{\min}^{(\epsilon_n)}e^{-n\epsilon_n}.
\end{equation}
Since by definition \(m_{\lambda}\geq D_{\min}^{(\epsilon_n)}\) for every
\(\lambda\) satisfying \(|H(\bar{\lambda})-S(\rho)|\leq\epsilon_n\) [see Eq.~\eqref{dimm_lambda_defn}], it follows that
\begin{equation}
\forall\,\lambda:\,|H(\bar{\lambda})-S(\rho)|\leq\epsilon_n,\qquad
\frac{\dim\left(\mathcal{P}_{\lambda}^{'\perp}\right)}{m_{\lambda}}
\leq
\frac{\dim\left(\mathcal{P}_{\lambda}^{'\perp}\right)}{D_{\min}^{(\epsilon_n)}}
\leq
e^{-n\epsilon_n},
\end{equation}
where the final inequality follows from Eq.~\eqref{remainderdimbound}.
\end{proof}
\begin{thm}\label{Deutsche_bahn}
    For any $\epsilon_n>0$ we have
    \begin{equation}
        \left\|
\rho^{\otimes n}
-
\frac{
P_{\epsilon_n}
\rho^{\otimes n}
P_{\epsilon_n}
}{
\Tr\left(
P_{\epsilon_n}
\rho^{\otimes n}
P_{\epsilon_n}
\right)
}
\right\|_1
\leq
\poly(n)e^{-n\eta_{\epsilon_n}/2}
    +
    2\sqrt{2}t\,e^{-n\alpha({\epsilon_n})}
    +2e^{-n\epsilon_n/2}\leq\poly(n)\exp\!\left(
-n\kappa_{\epsilon_n}
\right),
\end{equation}
where we denote
\begin{equation}\label{kappadefn}
    \kappa_{\epsilon_n} :=\min\!\left\{\frac{\eta_{\epsilon_n}}{2},\;\alpha({\epsilon_n}),\;\frac{\epsilon_n}{2}\right\}.
\end{equation}

\end{thm}
\begin{proof}
We can proceed as follows:
\begin{align}
        \left\|
\rho^{\otimes n}
-
\frac{
P_{\epsilon_n}
\rho^{\otimes n}
P_{\epsilon_n}
}{
\Tr\left(
P_{\epsilon_n}
\rho^{\otimes n}
P_{\epsilon_n}
\right)
}
\right\|_1&\overset{(1)}{=}\left\|
    \rho^{\otimes n}
    -
    \frac{
    \Pi_{\epsilon_n}\tilde{\Pi}_{\epsilon_n}
    \rho^{\otimes n}
    \tilde{\Pi}_{\epsilon_n}\Pi_{\epsilon_n}
    }{
    \Tr\left(
    \Pi_{\epsilon_n}\tilde{\Pi}_{\epsilon_n}
    \rho^{\otimes n}
    \tilde{\Pi}_{\epsilon_n}\Pi_{\epsilon_n}
    \right)
    }+\frac{
    \Pi_{\epsilon_n}\tilde{\Pi}_{\epsilon_n}
    \rho^{\otimes n}
    \tilde{\Pi}_{\epsilon_n}\Pi_{\epsilon_n}
    }{
    \Tr\left(
    \Pi_{\epsilon_n}\tilde{\Pi}_{\epsilon_n}
    \rho^{\otimes n}
    \tilde{\Pi}_{\epsilon_n}\Pi_{\epsilon_n}
    \right)
    }-\frac{
P_{\epsilon_n}
\rho^{\otimes n}
P_{\epsilon_n}
}{
\Tr\left(
P_{\epsilon_n}
\rho^{\otimes n}
P_{\epsilon_n}
\right)
}
    \right\|_1\nonumber\\
    &\overset{(2)}{\leq} \left\|
    \rho^{\otimes n}
    -
    \frac{
    \Pi_{\epsilon_n}\tilde{\Pi}_{\epsilon_n}
    \rho^{\otimes n}
    \tilde{\Pi}_{\epsilon_n}\Pi_{\epsilon_n}
    }{
    \Tr\left(
    \Pi_{\epsilon_n}\tilde{\Pi}_{\epsilon_n}
    \rho^{\otimes n}
    \tilde{\Pi}_{\epsilon_n}\Pi_{\epsilon_n}
    \right)
    }\right\|_1+\left\|\frac{
    \Pi_{\epsilon_n}\tilde{\Pi}_{\epsilon_n}
    \rho^{\otimes n}
    \tilde{\Pi}_{\epsilon_n}\Pi_{\epsilon_n}
    }{
    \Tr\left(
    \Pi_{\epsilon_n}\tilde{\Pi}_{\epsilon_n}
    \rho^{\otimes n}
    \tilde{\Pi}_{\epsilon_n}\Pi_{\epsilon_n}
    \right)
    }-\frac{
P_{\epsilon_n}
\rho^{\otimes n}
P_{\epsilon_n}
}{
\Tr\left(
P_{\epsilon_n}
\rho^{\otimes n}
P_{\epsilon_n}
\right)
}
    \right\|_1\nonumber\\
    &\overset{(3)}{\leq} 2\sqrt{
    \poly(n)e^{-n\eta_{\epsilon_n}}
    +
    2t\,e^{-2n\alpha({\epsilon_n})}
    }+2e^{-n\epsilon_n/2}\nonumber\\
    & \overset{(4)}{\leq} \poly(n)e^{-n\eta_{\epsilon_n}/2}
+
2\sqrt{2t}\,e^{-n\alpha({\epsilon_n})}
+
2e^{-n\epsilon_n/2}\nonumber\\
&\overset{(5)}{\leq} \poly(n)\exp\!\left(
-n\kappa_{\epsilon_n}
\right),
    \end{align}
where to write first equality $(\overset{(1)}{\leq})$, we add and substract $\frac{
P_{\epsilon_n}
\rho^{\otimes n}
P_{\epsilon_n}
}{
\Tr\left(
P_{\epsilon_n}
\rho^{\otimes n}
P_{\epsilon_n}
\right)
}$.
We obtain the inequality in the second line $(\overset{(2)}{\leq})$ by using the triangle inequality. 
The inequality in the third line $(\overset{(3)}{\leq})$ follows from Eqs.~\eqref{ref200} and \eqref{bdtrace}. 
Finally, the inequality in the fourth line $(\overset{(4)}{\leq})$ is obtained by using the relation 
$\sqrt{a+b}\leq \sqrt{a}+\sqrt{b}$.
\end{proof}
\subsection{Finding order of the transformation error}

\begin{prop}\label{order_of_error_transformation}
Let $\epsilon_n$ be any positive decreasing function of $n$ with 
\begin{equation}
    \lim_{n\to\infty}\epsilon_n=0.
\end{equation}
Then 
\begin{equation}\label{kex}
    \kappa_{\epsilon_n} :=\min\!\left\{\frac{\eta_{\epsilon_n}}{2},\;\alpha({\epsilon_n}),\;\frac{\epsilon_n}{2}\right\},
\end{equation}
introduced in Eq. \eqref{kappadefn} where
\begin{align}
    \eta_{\epsilon_n}
    &
    =
    \frac{\bigl(\xi^{-1}(\epsilon_n)\bigr)^2}{2},
    \qquad\mathrm{with}\qquad
    \xi(x)
    =
    x\log d + H_2(x),\qquad\mathrm{where}\qquad
    H_2(x)
    =
    -\bigl[x\log x +(1-x)\log(1-x)\bigr],\label{Binary_entropy}\\
    \alpha({\epsilon_n})
    &=
    \frac{2\,\epsilon_n^2}{\Delta L_{\max}^2},
    \qquad\mathrm{where}\qquad
    \Delta L_{\max}
    =
    \max_{c\in\{1,\ldots,t\}}
    \Bigl[
        \operatorname{Eig}_{\max}(L_c)
        -
        \operatorname{Eig}_{\min}(L_c)
    \Bigr],
\end{align}
satisfy 
\begin{equation}
    \kappa_{\epsilon_n}
    \sim
    \frac{1}{4}\left(\frac{\epsilon_n}{\log(1/\epsilon_n)}\right)^2.
\end{equation}
\end{prop}

\begin{proof}
 We proceed by showing that 
\begin{equation}
\xi^{-1}(\epsilon_n)
\sim\left(\frac{\epsilon_n}{\log(1/\epsilon_n)}\right),
\end{equation}
where $\xi^{-1}(\epsilon_n)$ denotes the solution of
\begin{equation}\label{above_eqn_xix}
    \epsilon_n=x\log d+H_2(x)=\xi(x). 
\end{equation}
This equivalently means 
\begin{equation}\label{defining_reln}
   \epsilon_n=\xi^{-1}(\epsilon_n)\log d
    +H_2\!\left(\xi^{-1}(\epsilon_n)\right).
\end{equation}
Now, observe that $\xi^{-1}(\epsilon_n)\to 0$ whenever
$\epsilon_n\to 0$. Indeed, both terms on the right-hand side of
Eqs.~\eqref{above_eqn_xix} are non-negative. Therefore, if
$\epsilon_n\to 0$, then according to Eq. \eqref{defining_reln} necessarily
\[
\xi^{-1}(\epsilon_n)\log d\to 0
\qquad\text{and}\qquad
H_2\left(\xi^{-1}(\epsilon_n)\right)\to 0.
\]
Since $\log d$ is a constant, the first relation implies $\xi^{-1}(\epsilon_n)\to 0$.
Hence, the solution $\xi^{-1}(\epsilon_n)$ vanishes in the limit
$\epsilon_n\to 0$. 
Substituting the definition of binary entropy from Eq. \eqref{Binary_entropy} we can write Eq. \eqref{defining_reln} as
\begin{equation}
    \epsilon_n=\xi^{-1}(\epsilon_n)\log\left(\frac{1}{\xi^{-1}(\epsilon_n)}\right)+\xi^{-1}(\epsilon_n)\log d+\left(1-\xi^{-1}(\epsilon_n)\right)\log\left(\frac{1}{1-\xi^{-1}(\epsilon_n)}\right)\label{ubseq2}.
\end{equation}

Next, we observe that for $0<\xi^{-1}(\epsilon_n)<1/2$, using identity $y\leq \ln\left(\frac{1}{1-y}\right)\leq \frac{y}{1-y}$ for $0<y<1$ we have
\begin{equation}\label{inequality_useful}
    \frac{1}{2}\left(\frac{\xi^{-1}(\epsilon_n)}{1-\xi^{-1}(\epsilon_n)}\right)\overset{(1)}{\leq} \xi^{-1}(\epsilon_n)\leq \log\left(\frac{1}{1-\xi^{-1}(\epsilon_n)}\right) \leq \frac{\xi^{-1}(\epsilon_n)}{1-\xi^{-1}(\epsilon_n)},
\end{equation}
where to write $(\overset{(1)}{\leq})$, we have used 
\begin{equation}\label{eqn_ab}
    \mathrm{for}\qquad 0<\xi^{-1}(\epsilon_n)<1/2\qquad \mathrm{we\;\;have} \qquad\xi^{-1}(\epsilon_n)= \left(1-\xi^{-1}(\epsilon_n)\right)\left(\frac{\xi^{-1}(\epsilon_n)}{1-\xi^{-1}(\epsilon_n)}\right)\geq \frac{1}{2}\left(\frac{\xi^{-1}(\epsilon_n)}{1-\xi^{-1}(\epsilon_n)}\right).
\end{equation}
Using the above Eq. \eqref{eqn_ab}, we obtain the following from Eq. \eqref{ubseq2}:
\begin{align}
    \xi^{-1}(\epsilon_n)\log\left(\frac{1}{\xi^{-1}(\epsilon_n)}\right)+\xi^{-1}(\epsilon_n)\left[\log d+\frac{1}{2}\right]\overset{(1)}{\leq}\epsilon_n&\overset{(2)}{\leq} \xi^{-1}(\epsilon_n)\log\left(\frac{1}{\xi^{-1}(\epsilon_n)}\right)+\xi^{-1}(\epsilon_n)\left[\log d+1\right].\label{ubseq},
\end{align}
To obtain the inequality $(\overset{(1)}{\leq})$ and $(\overset{(2)}{\leq})$ we use inequality from Eq.~\eqref{inequality_useful}. 

Note that $\xi^{-1}(\epsilon_n)\to 0$ as $n\to\infty$, since $\epsilon_n$ is a decreasing function of $n$.
Consequently, with $\xi^{-1}(\epsilon_n)\to 0$ we have
\begin{equation}
    \lim_{n\rightarrow\infty}\frac{1}{\log\!\left(\frac{1}{\xi^{-1}(\epsilon_n)}\right)}=0,
\end{equation}
which allows us to write 
\begin{equation}\label{95v}
   \xi^{-1}(\epsilon_n)\left[\log d+a\right] = o\left( \xi^{-1}(\epsilon_n)\log\left(\frac{1}{\xi^{-1}(\epsilon_n)}\right)\right),\quad\text{where $a$ is an arbitrary constant.}
\end{equation}
This implies that we can express Eq. \eqref{ubseq} using Eq. \eqref{95v},
\begin{equation}\label{95we}
    \xi^{-1}(\epsilon_n)\log\left(\frac{1}{\xi^{-1}(\epsilon_n)}\right)+o\left( \xi^{-1}(\epsilon_n)\log\left(\frac{1}{\xi^{-1}(\epsilon_n)}\right)\right)\leq\epsilon_n\leq \xi^{-1}(\epsilon_n)\log\left(\frac{1}{\xi^{-1}(\epsilon_n)}\right)+o\left( \xi^{-1}(\epsilon_n)\log\left(\frac{1}{\xi^{-1}(\epsilon_n)}\right)\right).
\end{equation}
Therefore, from Eq. \eqref{95we} we can have 
\begin{equation}\label{neg_log}
   \epsilon_n=\xi^{-1}(\epsilon_n)\log\left(\frac{1}{\xi^{-1}(\epsilon_n)}\right)+o\left( \xi^{-1}(\epsilon_n)\log\left(\frac{1}{\xi^{-1}(\epsilon_n)}\right)\right),
   \quad\text{in other words}quad
   \epsilon_n\sim \xi^{-1}(\epsilon_n)\log\left(\frac{1}{\xi^{-1}(\epsilon_n)}\right).
\end{equation}
Taking logarithm on the both side of Eq. \eqref{neg_log}, we have 
\begin{equation}\label{neg_log2}
    \log\left(\frac{1}{\epsilon_n}\right)\sim \log \left(\frac{1}{\xi^{-1}(\epsilon_n)}\right)-\log\left(\log\left(\frac{1}{\xi^{-1}(\epsilon_n)}\right)\right).
\end{equation}
As $\xi^{-1}(\epsilon_n)\to 0$ as $n\to\infty$, we see that the second term in Eq. \eqref{neg_log2} diverges much slower than the first term with $n\to\infty$, we have
\begin{equation}
    \log\left(\log\left(\frac{1}{\xi^{-1}(\epsilon_n)}\right)\right)=o\left(\log \left(\frac{1}{\xi^{-1}(\epsilon_n)}\right)\right),
\end{equation}
which equivalently means
\begin{equation}\label{equivalently_means}
    \log\left(\frac{1}{\epsilon_n}\right)\sim \log\left(\frac{1}{\xi^{-1}\left(\epsilon_n\right)}\right).
\end{equation}
Using Eq. \eqref{neg_log} we can write
\begin{equation}
    \xi^{-1}\left(\epsilon_n\right) \sim \frac{\epsilon_n}{\log\left(\frac{1}{\xi^{-1}\left(\epsilon_n\right)}\right)}\sim\left(\frac{\epsilon_n}{\log(1/\epsilon_n)}\right).
\end{equation}
where to write final equality we use Eq. \eqref{equivalently_means}. This further gives 
\begin{equation}
   \frac{\eta_{\epsilon_n}}{2} = \frac{\bigl(\xi^{-1}(\epsilon_n)\bigr)^2}{4} \sim \frac{1}{4}\left(\frac{\epsilon_n}{\log(1/\epsilon_n)}\right)^2,
\end{equation}
since the other terms appearing in Eq.~\eqref{kex},
\begin{equation}
      \alpha({\epsilon_n})= \frac{2\,\epsilon_n^2}{\Delta L_{\max}^2} \qquad\text{and}\qquad \frac{\epsilon_n}{2},
\end{equation}
are greater than $\frac{1}{4}\left(\frac{\epsilon_n}{\log(1/\epsilon_n)}\right)^2$ for sufficiently large $n$. This means 
  \begin{equation}
    \kappa_{\epsilon_n}
    =
    \min\left\{
        \frac{\eta_{\epsilon_n}}{2},
        \alpha({\epsilon_n}),
        \frac{\epsilon_n}{2}
    \right\} \sim \frac{1}{4}\left(\frac{\epsilon_n}{\log(1/\epsilon_n)}\right)^2,
\end{equation}
which completes the proof.
\end{proof}

\subsection{Constructing the charge-conserving unitary}
\begin{defn}[Charge conserving unitary on a subspace]
Let
\(\{Q_{\mathrm{in}}^{(i)}\}_{i=1}^{m}\) and
\(\{Q_{\mathrm{out}}^{(i)}\}_{i=1}^{m}\)
be sets of mutually commuting conserved charges associated with the Hilbert spaces
\(\mathcal{H}_{\mathrm{in}}\) and
\(\mathcal{H}_{\mathrm{out}}\), respectively.
Let \(P\) be the projector onto a subspace of
\(\mathcal{H}_{\mathrm{in}}\).
A unitary
\[
V:\operatorname{supp}(P)\rightarrow\mathcal{H}_{\mathrm{out}}
\]
is said to be \emph{charge conserving} if
\begin{equation}
    V^{\dagger} Q_{\mathrm{out}}^{(i)} V
    =
    P Q_{\mathrm{in}}^{(i)} P,
    \qquad
    i=1,\ldots,m .
\end{equation}
If furthermore the projector \(P\) commutes with every input charge,
i.e.,
\[
[P,Q_{\mathrm{in}}^{(i)}]=0,
\qquad
i=1,\ldots,m,
\]
then the charge-conservation condition can equivalently be written as
\begin{equation}\label{definition_charge_conserving_unitary_on_subspace}
    V^{\dagger} Q_{\mathrm{out}}^{(i)} V
    =
    P Q_{\mathrm{in}}^{(i)},
    \qquad
    i=1,\ldots,m .
\end{equation}
\end{defn}

\begin{thm}
\label{U_lambda_defn_thm9232}
For any \(\epsilon_n>0\), consider, for each \(\lambda\) satisfying 
\(|H(\bar\lambda)-S(\rho)|\leq\epsilon_n\), a unitary
\begin{equation}\label{U_lambda_defn_thm}
U_\lambda:\mathcal{P}'_\lambda\rightarrow\mathcal{E}\otimes\mathcal{F}_\lambda,
\end{equation}
where \(\mathcal{E}\) and \(\mathcal{F}_\lambda\) are chosen such that
\(\mathcal{P}'_\lambda\simeq\mathcal{E}\otimes\mathcal{F}_\lambda\) [see Eq.~\eqref{dimm_prime_lambda}], with
\begin{equation}\label{EdimFdim}
\dim (\mathcal{E})=\left\lfloor D_{\min}^{(\epsilon_n)}e^{-n\epsilon_n}\right\rfloor,\qquad 
\dim (\mathcal{F}_\lambda)=\left\lfloor\frac{e^{n\epsilon_n}m_\lambda}{D_{\min}^{(\epsilon_n)}}\right\rfloor .
\end{equation}
The above unitaries can be extended to a charge-conserving unitary over the support of \(P_{\epsilon_n}\) introduced in Eq. \eqref{Pen} as
\begin{equation}\label{seq_eqn}
U=
\bigoplus_{\lambda:\,|H(\bar\lambda)-S(\rho)|\leq\epsilon_n}
U_\lambda\otimes\1^{(\mathcal{Q}'_\lambda)}
:
\operatorname{supp}(P_{\epsilon_n})
=
\bigoplus_{\lambda:\,|H(\bar\lambda)-S(\rho)|\leq\epsilon_n}
\mathcal{P}'_\lambda\otimes\mathcal{Q}'_\lambda
\rightarrow
\mathcal{E}\otimes\mathcal{K},
\quad\mathrm{where}\quad
\mathcal{K}
:=
\bigoplus_{\lambda:\,|H(\bar\lambda)-S(\rho)|\leq\epsilon_n}
\mathcal{F}_\lambda\otimes\mathcal{Q}'_\lambda .
\end{equation}
with the following choice of charges on \(\mathcal{E}\) and \(\mathcal{K}\):
\begin{equation}\label{hypothetical_charge}
L^{(\mathcal{E})}_{c}=nl_c\,\1^{(\mathcal{E})},\qquad 
L^{(\mathcal{K})}_{c}
=
\bigoplus_{\lambda:\,|H(\bar\lambda)-S(\rho)|\leq\epsilon_n}
\1^{(\mathcal{F}_\lambda)}\otimes
\left(
\tilde L^{(\mathcal{Q}'_\lambda)}_{\lambda,c}
-nl_c\,\1^{(\mathcal{Q}'_\lambda)}
\right),
\qquad
\tilde L^{(\mathcal{Q}'_\lambda)}_{\lambda,c}
=
\1^{(\mathcal{Q}'_\lambda)}_{\lambda}
L^{(\mathcal{Q}_\lambda)}_{\lambda,c}.
\end{equation}
Furthermore, the action of the unitary on the projected state is given by
\begin{equation}\label{eqn_projected_state}
    U \left(
    \frac{P_{\epsilon_n}\rho^{\otimes n}P_{\epsilon_n}}
    {\Tr(P_{\epsilon_n}\rho^{\otimes n}P_{\epsilon_n})}
    \right) U^{\dagger}
    =
    \frac{\1^{(\mathcal{E})}}{\dim (\mathcal{E})}
    \otimes
    \eta_{\rho}^{(\mathcal{K})},
\end{equation}
where $\eta_{\rho}^{(\mathcal{K})}$ is a state supported on Hilbert space $\mathcal{K}$
\end{thm}
\begin{proof}
    As the projector $P_{\epsilon_n}$ is supported on the subspace spanned by the common eigenvectors associated with the commuting charges [see Eq.~\eqref{Pen2}], it commutes with each conserved charge, i.e.,
\begin{equation}
[P_{\epsilon_n},L_c^{\times n}] = 0,\qquad \mathrm{for}\qquad c\in\{1,\ldots,t\}.
\end{equation} 
We first show that $U$ is charge-conserving unitary over the support of $P_{\epsilon_n}$ via checking Eq. \eqref{definition_charge_conserving_unitary_on_subspace} i.e.,
    \begin{equation}
        U^{\dagger}\left(L^{(\E)}_{c}\otimes\1^{(\mathcal{K})}+\1^{(\mathcal{E})}\otimes L^{(\mathcal{K})}_{c}\right) =P_{\epsilon_n} L_{c,n}\qquad \mathrm{for}\qquad c\in\{1,\ldots,t\}.
    \end{equation}
     To do so, we proceed by substituting $L^{(\E)}_{c}$ and $L^{(\mathcal{K})}_{c}$ from Eq. \eqref{hypothetical_charge} to calculate
\begin{align}\label{Energy_conservation_semi_app}
&U^{\dagger}\left(L^{(\E)}_{c}\otimes\1^{(\mathcal{K})}+\1^{(\mathcal{E})}\otimes L^{(\mathcal{K})}_{c}\right) U=\left[\bigoplus_{\lambda:|H(\bar\lambda)-S(\rho)|\leq \epsilon_n} U_{\lambda}^{\dagger} \otimes \1_{\lambda}^{\left(\mathcal{Q}'_{\lambda}\right)} \right]\left(nl_c\1^{\left(\mathcal{E}\right)}\otimes\1^{(\mathcal{K})}+\1^{(\mathcal{E})}\otimes L^{(\mathcal{K})}_{c}\right) \left[\bigoplus_{\lambda:|H(\bar\lambda)-S(\rho)|\leq \epsilon_n} U_{\lambda} \otimes \1_{\lambda}^{\left(\mathcal{Q}'_{\lambda}\right)} \right]\nonumber\\
    \overset{(1)}{=}& \left[\bigoplus_{\lambda:|H(\bar\lambda)-S(\rho)|\leq \epsilon_n} U_{\lambda}^{\dagger} \otimes \1_{\lambda}^{\left(\mathcal{Q}'_{\lambda}\right)} \right]\left(nl_c\1^{\left(\mathcal{E}\right)}\otimes\left[\bigoplus_{\lambda:|H(\bar\lambda)-S(\rho)|\leq \epsilon_n}\1_{\lambda}^{\left(\mathcal{F}_{\lambda}\right)} \otimes \1^{\left(\mathcal{Q}'_{\lambda}\right)}_{\lambda}\right]+\1^{\left(\mathcal{E}\right)}\otimes \left[\bigoplus_{\lambda:|H(\bar\lambda)-S(\rho)|\leq \epsilon_n}\1_{\lambda}^{\left(\mathcal{F}_{\lambda}\right)} \otimes \left(\tilde{L}^{\left(\mathcal{Q}'_{\lambda}\right)}_{\lambda, c}-nl_c\1^{\left(\mathcal{Q}'_{\lambda}\right)}_{\lambda}\right) \right] \right)\nonumber\\\times&\left[\bigoplus_{\lambda:|H(\bar\lambda)-S(\rho)|\leq \epsilon_n} U_{\lambda}\otimes \1_{\lambda}^{\left(\mathcal{Q}'_{\lambda}\right)} \right]\nonumber\\
     \overset{(2)}{=}& \bigoplus_{\lambda:|H(\bar\lambda)-S(\rho)|\leq \epsilon_n} U_{\lambda}^{\dagger}\left(\1^{(\E)}\otimes \1_{\lambda}^{\left(\mathcal{F}_{\lambda}\right)}\right) U_{\lambda}\otimes \tilde{L}^{\left(\mathcal{Q}'_{\lambda}\right)}_{\lambda, c} \overset{(3)}{=} \bigoplus_{\lambda:|H(\bar\lambda)-S(\rho)|\leq \epsilon_n} \1^{\left(\mathcal{P}'_{\lambda}\right)}_{\lambda}\otimes \tilde{L}^{\left(\mathcal{Q}'_{\lambda}\right)}_{\lambda, c}\nonumber\\\overset{(4)}{=} &\left[\bigoplus_{\lambda:|H(\bar\lambda)-S(\rho)|\leq \epsilon_n} \1^{\left(\mathcal{P}'_{\lambda}\right)}_{\lambda}\otimes \1^{\left(\mathcal{Q}'_{\lambda}\right)}_{\lambda}\right]\left[\bigoplus_{\lambda \in Y^n_d} 
    \1^{(\mathcal{P}_{\lambda})} \otimes 
L_{\lambda,c}^{(\mathcal{Q}_{\lambda})}\right]\nonumber\\\overset{(5)}{=}&P_{\epsilon_n} L_{c}^{\times n},
\end{align}
We obtain the first equality $(\overset{(1)}{=})$ by substituting [see Eq. \eqref{seq_eqn}]
\begin{equation}
\1^{(\mathcal{K})}
=
\bigoplus_{\lambda:|H(\bar\lambda)-S(\rho)|\leq \epsilon_n}
\1^{\left(\mathcal{F}_{\lambda}\right)}_{\lambda}
\otimes
\1^{\left(\mathcal{Q}'_{\lambda}\right)}_{\lambda}.
\end{equation}
The second equality $(\overset{(2)}{=})$ follows by substituting the definition of
$\tilde L^{(\mathcal{Q}'_\lambda)}_{\lambda,c}$ from
Eq.~\eqref{hypothetical_charge}, while the third equality
$(\overset{(3)}{=})$ follows from Eq. \eqref{U_lambda_defn_thm} which implies
\begin{equation}\label{UEKPen}
U_{\lambda}^{\dagger}
\left(
\1^{(\mathcal{E})}
\otimes
\1_{\lambda}^{\left(\mathcal{F}_{\lambda}\right)}
\right)
U_{\lambda}
=
\1^{\left(\mathcal{P}'_{\lambda}\right)}_{\lambda}.
\end{equation}
Next, we will prove Eq. \eqref{eqn_projected_state}. From Eq.~\eqref{Pen}, Eq.~\eqref{calculate}, together with Eq.~\eqref{up}, we obtain
\begin{equation}\label{projected_useful_state}
    \frac{
    P_{\epsilon_n}
    \rho^{\otimes n}
    P_{\epsilon_n}
    }{
    \Tr\left(
    P_{\epsilon_n}
    \rho^{\otimes n}
    P_{\epsilon_n}
    \right)
    }
    =
    \bigoplus_{\lambda:|H(\bar\lambda)-S(\rho)|\leq \epsilon_n}
    \tilde{r}_{\lambda}
    \frac{\1^{(\P'_{\lambda})}}{m'_{\lambda}}
    \otimes
    \tilde{\rho}_{\lambda}^{(\Q'_{\lambda})},
\end{equation}
where $\tilde{r}_{\lambda}$ is defined in Eq.~\eqref{normalization} and [see Eq.~\eqref{dimm_prime_lambda}]
\begin{equation}
\Tr\!\left(\1^{(\mathcal{P}'_{\lambda})}\right)
=
\dim (\mathcal{P}'_{\lambda})
=
m'_{\lambda}.
\end{equation}
Furthermore, the unitarity of $U_{\lambda}$ implies, upon taking the trace on both sides of Eq.~\eqref{UEKPen}, that
\begin{equation}\label{factor_dimension}
    \dim (\P'_{\lambda})
    =
    m'_{\lambda}
    =
    \dim (\E)\,
    \dim (\mathcal{F}_{\lambda}).
\end{equation}
Hence,
\begin{align}
     U \left(
    \frac{P_{\epsilon_n}\rho^{\otimes n}P_{\epsilon_n}}
    {\Tr(P_{\epsilon_n}\rho^{\otimes n}P_{\epsilon_n})}
    \right)
    U^{\dagger}
    &\overset{(1)}{=}
    \left(
    \bigoplus_{\lambda:\,|H(\bar\lambda)-S(\rho)|\leq\epsilon_n}
    U_\lambda\otimes\1^{(\mathcal{Q}'_\lambda)}
    \right)
    \left(
    \bigoplus_{\lambda:|H(\bar\lambda)-S(\rho)|\leq \epsilon_n}
    \tilde{r}_{\lambda}
    \frac{\1^{(\P'_{\lambda})}}{m'_{\lambda}}
    \otimes
    \tilde{\rho}_{\lambda}^{(\Q'_{\lambda})}
    \right)
    \left(
    \bigoplus_{\lambda:\,|H(\bar\lambda)-S(\rho)|\leq\epsilon_n}
    U_\lambda^{\dagger}\otimes\1^{(\mathcal{Q}'_\lambda)}
    \right)\nonumber\\
    &\overset{(2)}{=}
    \bigoplus_{\lambda:|H(\bar\lambda)-S(\rho)|\leq \epsilon_n}
    \tilde{r}_{\lambda}
    \frac{
    U_{\lambda}\1^{(\P'_{\lambda})}U_{\lambda}^{\dagger}
    }{
    \dim (\E)\dim (\mathcal{F}_{\lambda})
    }
    \otimes
    \tilde{\rho}_{\lambda}^{(\Q'_{\lambda})}\nonumber\\
    &\overset{(3)}{=}
    \frac{\1^{(\mathcal{E})}}{\dim (\E)}
    \otimes
    \underbrace{\left(
    \bigoplus_{\lambda:|H(\bar\lambda)-S(\rho)|\leq \epsilon_n}
    \tilde{r}_{\lambda}
    \frac{\1_{\lambda}^{(\mathcal{F}_{\lambda})}}
    {\dim (\mathcal{F}_{\lambda})}
    \otimes
    \tilde{\rho}_{\lambda}^{(\Q'_{\lambda})}
    \right)}_{:=\eta^{(\mathcal{K})}_{\rho}}
    \overset{(4)}{=}
    \frac{\1^{(\mathcal{E})}}{\dim (\E)}
    \otimes
    \eta^{(\mathcal{K})}_{\rho},
\end{align}
where the first equality $(\overset{(1)}{=})$ follows by substituting the expression for
$U$ from Eq.~\eqref{seq_eqn} and the projected state from
Eq.~\eqref{projected_useful_state}. The second equality
$(\overset{(2)}{=})$ follows from Eq.
\eqref{factor_dimension}. The third equality
$(\overset{(3)}{=})$ follows from Eq.~\eqref{UEKPen}, namely,
\[
U_{\lambda}\1^{(\mathcal{P}'_{\lambda})}U_{\lambda}^{\dagger}
=
\1^{(\mathcal{E})}
\otimes
\1^{(\mathcal{F}_{\lambda})},
\]
and the final equality $(\overset{(4)}{=})$ follows by taking
\begin{equation}
    \eta^{(\mathcal{K})}_{\rho}
    =
    \bigoplus_{\lambda:|H(\bar\lambda)-S(\rho)|\leq \epsilon_n}
    \tilde{r}_{\lambda}
    \frac{\1_{\lambda}^{(\mathcal{F}_{\lambda})}}
    {\dim (\mathcal{F}_{\lambda})}
    \otimes
    \tilde{\rho}_{\lambda}^{(\Q'_{\lambda})}.
\end{equation}
\end{proof}

\begin{thm}\label{charge_on_anc}
For any \(\epsilon_n>0\), the dimension of the Hilbert space \(\mathcal{K}\) and the corresponding charges \(L_c^{(\mathcal{K})}\) satisfy
\begin{equation}
    \dim (\mathcal{K})\leq \poly(n)e^{3n\epsilon_n},
    \qquad
    \|L_c^{(\mathcal{K})}\|\leq n\epsilon_n,
    \qquad c\in\{1,\ldots,t\},
\end{equation}
where \(\|\cdot\|\) denotes the operator norm.
\end{thm}
\begin{proof}
From the definition
\begin{equation}
    \mathcal{K}
    :=
    \bigoplus_{\lambda:\,|H(\bar\lambda)-S(\rho)|\leq\epsilon_n}
    \mathcal{F}_\lambda\otimes\mathcal{Q}'_\lambda,
\end{equation}
we obtain
\begin{equation}
    \dim (\mathcal{K})
    =
    \sum_{\lambda:\,|H(\bar\lambda)-S(\rho)|\leq\epsilon_n}
    \dim (\mathcal{F}_{\lambda})
    \dim (\Q'_{\lambda})
    =
    \sum_{\lambda:\,|H(\bar\lambda)-S(\rho)|\leq\epsilon_n}
    \left\lfloor\frac{e^{n\epsilon_n}m_\lambda}{D_{\min}^{(\epsilon_n)}}\right\rfloor
    \dim (\Q'_{\lambda})
    \leq
    \mathrm{poly}(n)e^{n\epsilon_n}
    \frac{D_{\max}^{(\epsilon_n)}}{D_{\min}^{(\epsilon_n)}},
\end{equation}
where we use the facts that $\dim (\Q'_{\lambda})=\mathrm{poly}(n)$ and that the cardinality of the set
$\{\lambda:\, |H(\bar{\lambda})-S(\rho)|\leq \epsilon_n\}$ is also $\mathrm{poly}(n)$. Here, using Eq. \eqref{dimplql}, we define 
\begin{equation}\label{rate_eqn_lambda}
    D_{\max}^{(\epsilon_n)}
    :=
    \max_{\lambda:\, |H(\bar{\lambda})-S(\rho)|\leq \epsilon_n}
    m_{\lambda},
    \qquad\text{where}\qquad
    m_{\lambda}
    =
    \dim (\P_{\lambda}).
\end{equation}

Next, we employ the following bounds from \cite[{Eqs.~(6.15)-(6.16)}]{harrow2005applications}:
\begin{align}
  \label{ineq_Harrow}
    \forall \lambda:\, |H(\bar\lambda)-S(\rho)|\leq \epsilon_n, \quad
    \frac{1}{\poly(n)}\,e^{n(S(\rho)-\epsilon_n)}
    \leq
    \frac{1}{\poly(n)}\,e^{nH(\bar{\lambda})}
    \leq
    \dim (\mathcal{P}_{\lambda}) = m_{\lambda}
    \leq
    e^{nH(\bar{\lambda})}
    \leq
    e^{n(S(\rho)+\epsilon_n)}.
\end{align}
It follows immediately that
\begin{equation}
    \frac{D_{\max}^{(\epsilon_n)}}{D_{\min}^{(\epsilon_n)}}
    \leq
    \mathrm{poly}(n)e^{2n\epsilon_n},
\end{equation}
and therefore
\begin{equation}
    \dim (\mathcal{K})
    \leq
    \mathrm{poly}(n)e^{n\epsilon_n}
    \frac{D_{\max}^{(\epsilon_n)}}{D_{\min}^{(\epsilon_n)}}
    \leq
    \mathrm{poly}(n)e^{3n\epsilon_n}.
\end{equation}
Next, we show that
\begin{equation}
    \|L_c^{(\mathcal{K})}\|
    \leq
    n\epsilon_n,
    \qquad
    c\in\{1,\ldots,t\}.
\end{equation}
Indeed,
\begin{align}
    L^{(\mathcal{K})}_{c}
    &=
    \bigoplus_{\lambda:\,|H(\bar\lambda)-S(\rho)|\leq\epsilon_n}
    \1^{(\mathcal{F}_\lambda)}
    \otimes
    \left(
    \tilde L^{(\mathcal{Q}'_\lambda)}_{\lambda,c}
    -
    nl_c\,\1^{(\mathcal{Q}'_\lambda)}
    \right)\qquad\mathrm{where}\qquad \tilde L^{(\mathcal{Q}'_\lambda)}_{\lambda,c}
    =
    \1^{(\mathcal{Q}'_\lambda)}_{\lambda}
    L^{(\mathcal{Q}_\lambda)}_{\lambda,c}.
\end{align}
Equivalently,
\begin{align}
    L^{(\mathcal{K})}_{c}
    =
    \bigoplus_{\lambda:\,|H(\bar\lambda)-S(\rho)|\leq\epsilon_n}
    \left(
    \1^{(\mathcal{F}_\lambda)}
    \otimes
    \1^{(\mathcal{Q}'_\lambda)}
    \right)
    \left(
    L^{(\mathcal{Q}_\lambda)}_{\lambda,c}
    -
    nl_c\,\1^{(\mathcal{Q}_\lambda)}
    \right).
\end{align}
Since the projection is onto the typical charge subspace defined in
Eq.~\eqref{energy_typ_proj}, Eq.~\eqref{Piepsilonn} implies that
\begin{equation}
    \|L_c^{(\mathcal{K})}\|
    \leq
    n\epsilon_n,
    \qquad
    c\in\{1,\ldots,t\},
\end{equation}
which completes the proof.
\end{proof}
\begin{thm}\label{dimEbd}
 For any \(\epsilon_n>0\), the dimension of the Hilbert space $\E$ introduced in Eq.~\eqref{U_lambda_defn_thm9232} satisfies
\begin{equation}
    \frac{1}{\poly(n)}e^{n(S(\rho)-2\epsilon_n)}-1
    \leq
    \dim (\mathcal{E})
    \leq
    e^{nS(\rho)}.
\end{equation}
\end{thm}
\begin{proof}
From Eq.~\eqref{EdimFdim}, we have
\begin{equation}
    \frac{1}{\poly(n)}e^{n(S(\rho)-2\epsilon_n)}-1
    \leq
    D_{\min}^{(\epsilon_n)}e^{-n\epsilon_n}-1
    \overset{(1)}{\leq}
    \dim (\mathcal{E})
    =
    \left\lfloor D_{\min}^{(\epsilon_n)}e^{-n\epsilon_n}\right\rfloor
    \overset{(2)}{\leq}
    D_{\min}^{(\epsilon_n)}e^{-n\epsilon_n}
    \leq
    e^{nS(\rho)},
\end{equation}
where the inequalities \(\overset{(1)}{\leq}\) and \(\overset{(2)}{\leq}\) follow from the floor-function bound
\(x-1\leq \lfloor x\rfloor\leq x\), and the remaining inequalities follow from Eq.~\eqref{ineq_Harrow}.
\end{proof}
\subsection{Proof of Theorem \ref{Main_Central_Thm}}\label{cc1_calc}

    Consider any $\rho\in\mathcal{C}(s,l_1,\ldots,l_t)$. Using the projector $P_{\epsilon_n}$ introduced in Eq.~\eqref{Pen} together with the Hilbert spaces $\mathcal{E}$ and $\mathcal{K}$ defined in Theorem~\ref{U_lambda_defn_thm9232} [see Eqs.~\eqref{U_lambda_defn_thm}--\eqref{EdimFdim}], we obtain, for any unitary $V$,
    \begin{align}
    \left\|
    V\rho^{\otimes n}V^\dagger
    -
    \frac{\1^{(\mathcal{E})}}{\dim (\mathcal{E})}
    \otimes
    \eta_{\rho}^{(\mathcal{K})}
    \right\|_1&=\left\|
    V\rho^{\otimes n}V^\dagger
    -V\left(
    \frac{P_{\epsilon_n}\rho^{\otimes n}P_{\epsilon_n}}
    {\Tr(P_{\epsilon_n}\rho^{\otimes n}P_{\epsilon_n})}
    \right)V^\dagger+V\left(
    \frac{P_{\epsilon_n}\rho^{\otimes n}P_{\epsilon_n}}
    {\Tr(P_{\epsilon_n}\rho^{\otimes n}P_{\epsilon_n})}
    \right)V^\dagger-
    \frac{\1^{(\mathcal{E})}}{\dim (\mathcal{E})}
    \otimes
    \eta_{\rho}^{(\mathcal{K})}
    \right\|_1\nonumber\\&\overset{(1)}{\leq} \left\|
    V\rho^{\otimes n}V^\dagger
    -V\left(
    \frac{P_{\epsilon_n}\rho^{\otimes n}P_{\epsilon_n}}
    {\Tr(P_{\epsilon_n}\rho^{\otimes n}P_{\epsilon_n})}
    \right)V^\dagger\right\|_{1}+\left\|V\left(
    \frac{P_{\epsilon_n}\rho^{\otimes n}P_{\epsilon_n}}
    {\Tr(P_{\epsilon_n}\rho^{\otimes n}P_{\epsilon_n})}
    \right)V^\dagger-
    \frac{\1^{(\mathcal{E})}}{\dim (\mathcal{E})}
    \otimes
    \eta_{\rho}^{(\mathcal{K})}
    \right\|_1\nonumber\\ &\overset{(2)}{=}\left\|
    \rho^{\otimes n}
    -\left(
    \frac{P_{\epsilon_n}\rho^{\otimes n}P_{\epsilon_n}}
    {\Tr(P_{\epsilon_n}\rho^{\otimes n}P_{\epsilon_n})}
    \right)\right\|_{1}+\left\|V\left(
    \frac{P_{\epsilon_n}\rho^{\otimes n}P_{\epsilon_n}}
    {\Tr(P_{\epsilon_n}\rho^{\otimes n}P_{\epsilon_n})}
    \right)V^\dagger-
    \frac{\1^{(\mathcal{E})}}{\dim (\mathcal{E})}
    \otimes
    \eta_{\rho}^{(\mathcal{K})}
    \right\|_1\\
    &\overset{(3)}{\leq} \poly(n)\exp\!\left(
-n\kappa_{\epsilon_n}
\right)+ \left\|V\left(
    \frac{P_{\epsilon_n}\rho^{\otimes n}P_{\epsilon_n}}
    {\Tr(P_{\epsilon_n}\rho^{\otimes n}P_{\epsilon_n})}
    \right)V^\dagger-
    \frac{\1^{(\mathcal{E})}}{\dim (\mathcal{E})}
    \otimes
    \eta_{\rho}^{(\mathcal{K})}
    \right\|_1\nonumber\\
    &\overset{(4)}{\lesssim} \poly(n)\exp\!\left(
-\frac{n}{4}
\left(
\frac{\epsilon_n}{\log(1/\epsilon_n)}
\right)^2
\right)+\left\|V\left(
    \frac{P_{\epsilon_n}\rho^{\otimes n}P_{\epsilon_n}}
    {\Tr(P_{\epsilon_n}\rho^{\otimes n}P_{\epsilon_n})}
    \right)V^\dagger-
    \frac{\1^{(\mathcal{E})}}{\dim (\mathcal{E})}
    \otimes
    \eta_{\rho}^{(\mathcal{K})}
    \right\|_1\label{sec_term},
    \end{align}
where the inequality $(\overset{(1)}{\leq})$ follows from the triangle inequality, the equality $(\overset{(2)}{=})$ follows from the unitary invariance of the trace norm, the inequality $(\overset{(3)}{\leq})$ follows from Theorem \ref{Deutsche_bahn} and the inequality $(\overset{(4)}{\lesssim})$ follows from Proposition \ref{order_of_error_transformation}. Since from Theorem \ref{thm_8_important} we have
\begin{equation}
    \operatorname{supp}(P_{\epsilon_n})
\subseteq
\operatorname{supp}\left(\Pi_{\epsilon_n}\tilde{\Pi}_{\epsilon_n}\right)\subseteq \mathcal{H}_d^{\otimes n}\qquad\text{which implies}\qquad \mathcal{H}_d^{\otimes n} = \operatorname{supp}(P_{\epsilon_n})+\left(\operatorname{supp}(P_{\epsilon_n})\right)^{\perp}.
\end{equation}
Choose
\begin{equation}\label{Vdef}
    V:\mathcal{H}_d^{\otimes n}
    =
    \operatorname{supp}(P_{\epsilon_n})
    \oplus
    \left(\operatorname{supp}(P_{\epsilon_n})\right)^{\perp}
    \longrightarrow
    (\mathcal{E}\otimes\mathcal{K})
    \oplus
    \left(\operatorname{supp}(P_{\epsilon_n})\right)^{\perp},
    \qquad
    V:=U\oplus\1^{(\operatorname{supp}(P_{\epsilon_n})^\perp)},
\end{equation}
where $U$ is the unitary of the form given in Eq.~\eqref{seq_eqn} of Theorem~\ref{U_lambda_defn_thm9232}, acting on $\operatorname{supp}(P_{\epsilon_n})$, while $\1^{(\operatorname{supp}(P_{\epsilon_n})^\perp)}$ denotes the identity on the orthogonal complement $\operatorname{supp}(P_{\epsilon_n})^{\perp}$. Thus, \(V\) extends \(U\) to the entire space by acting as the identity on \(\left(\operatorname{supp}(P_{\epsilon_n})\right)^{\perp}\). This implies the second term in Eq. \eqref{sec_term} upon substituting $V$ from Eq. \eqref{Vdef} simplifies to
\begin{align}\label{mix}
    \left\|V\left(
    \frac{P_{\epsilon_n}\rho^{\otimes n}P_{\epsilon_n}}
    {\Tr(P_{\epsilon_n}\rho^{\otimes n}P_{\epsilon_n})}
    \right)V^\dagger-
    \frac{\1^{(\mathcal{E})}}{\dim (\mathcal{E})}
    \otimes
    \eta_{\rho}^{(\mathcal{K})}
    \right\|_1(\overset{(1)}{=}) \left\|U\left(
    \frac{P_{\epsilon_n}\rho^{\otimes n}P_{\epsilon_n}}
    {\Tr(P_{\epsilon_n}\rho^{\otimes n}P_{\epsilon_n})}
    \right)U^\dagger-
    \frac{\1^{(\mathcal{E})}}{\dim (\mathcal{E})}
    \otimes
    \eta_{\rho}^{(\mathcal{K})}
    \right\|_1\overset{(2)}{=}\;0
    \end{align}
   where $(\overset{(1)}{=})$ follows from the fact that 
$P_{\epsilon_n}\rho^{\otimes n}P_{\epsilon_n}$ is supported on 
$\operatorname{supp}(P_{\epsilon_n})$, together with the block-diagonal structure of 
$V$ given in Eq.~\eqref{Vdef}. The equality $(\overset{(2)}{=})$ of Eq. \eqref{mix} follows from 
Eq.~\eqref{eqn_projected_state} of Theorem~\ref{U_lambda_defn_thm9232}. 
Next, we choose the charges on $\mathcal{E}$ and $\mathcal{K}$ to be 
$L_c^{(\mathcal{E})}$ and $L_c^{(\mathcal{K})}$, respectively, as defined in 
Eq.~\eqref{hypothetical_charge} (see Theorem~\ref{U_lambda_defn_thm9232}). 
From Theorems~\ref{U_lambda_defn_thm9232}, \ref{charge_on_anc} and \ref{dimEbd}, we have already established that
\begin{align}
    \frac{1}{\poly(n)}e^{n(s-2\epsilon_n)}-1
    &\leq
    \dim (\mathcal{E})
    \leq e^{ns},
    &
    L_c^{(\mathcal{E})}
    &=nl_c\,\1^{(\mathcal{E})},
    \\
    \dim (\mathcal{K})
    &\leq \poly(n)e^{3n\epsilon_n},
    &
    \|L_c^{(\mathcal{K})}\|
    &\leq n\epsilon_n ,
\end{align}
Finally, to check charge conservation, we aim to show
\begin{equation}\label{B114}
    V^{\dagger}
\left[
\left(
L_c^{(\mathcal{E})}\otimes\1^{(\mathcal{K})}
+
\1^{(\mathcal{E})}\otimes L_c^{(\mathcal{K})}
\right)
\oplus
P_{\epsilon_n}^{\perp}L_c^{\times n}
\right]
V
=
(U^{\dagger}\oplus\1)
\left[
\left(
L_c^{(\mathcal{E})}\otimes\1^{(\mathcal{K})}
+
\1^{(\mathcal{E})}\otimes L_c^{(\mathcal{K})}
\right)
\oplus
P_{\epsilon_n}^{\perp}L_c^{\times n}
\right]
(U\oplus\1)
=
L_c^{\times n},
\end{equation}
where first equality follows from the choice of $V$ as taken in Eq. \eqref{Vdef}. As $V$ acts trivially on $\left(\operatorname{supp}(P_{\epsilon_n})\right)^{\perp}$ [see Eq.~\eqref{Vdef}], to check charge conservation, it is sufficient to verify
\begin{equation}
\label{verify1}
    U^{\dagger}
\left(
L_c^{(\mathcal{E})}\otimes\1^{(\mathcal{K})}
+
\1^{(\mathcal{E})}\otimes L_c^{(\mathcal{K})}
\right)
U
=
P_{\epsilon_n}L_c^{\times n}.
\end{equation}
This relation has in fact already been established in Eq.~\eqref{Energy_conservation_semi_app} in the proof of Theorem~\ref{U_lambda_defn_thm9232}.

\section{State transformation from an unknown initial state}
Recall from the main text that the asymptotic transformation between known initial states in a $d$-dimensional system, prepared independently and identically, is completely characterized by the entropy and the average conserved charges of the states. Specifically, a transformation from $\rho$ to $\sigma$ is achievable asymptotically if and only if
\begin{equation}
    S(\rho)=S(\sigma)=s
    \quad \text{and} \quad
    \Tr(L_c\rho)=\Tr(L_c\sigma) = l_c
    \qquad \forall\, c\in\{1,\ldots,t\}.
\end{equation}
Thus, to achieve universal state transformation from an unknown initial state $\rho$ to a known target state $\sigma$, our objective is to efficiently estimate the von Neumann entropy $S(\rho)$ and the average charges $\Tr(L_c\rho)$ for all $c\in\{1,\ldots,t\}$ using $n$ independently and identically prepared copies of $\rho$, while ensuring that the disturbance caused to the state vanishes asymptotically. We emphasize that the only promise about the unknown state is that it is supported on a $d$-dimensional Hilbert space, i.e., $\rho \in \mathcal{B}(\mathcal{H}_d)$, and that we are given $n$ independently and identically prepared copies of $\rho$.

\subsection{Estimating the entropy of an unknown state}
We begin by estimating the entropy of the unknown state $\rho \in \mathcal{B}(\mathcal{H}_d)$ given $n$ identical and independently prepared copies of it. Recall from the Schur--Weyl decomposition in Eq.~\eqref{Schur_Weyl_Hilbert_space} that the $n$-copy Hilbert space $\mathcal{H}_d^{\otimes n}$ admits the decomposition
\begin{equation}
    \mathcal{H}_d^{\otimes n}
    =
    \bigoplus_{\lambda \in Y_d^n}
    \mathcal{P}_{\lambda} \otimes \mathcal{Q}_{\lambda},
\end{equation}
where the direct sum runs over all Young diagrams $\lambda \in Y_d^n$, which correspond to ordered partitions of the integer $n$ into at most $d$ non-increasing parts. Explicitly, each partition is of the form $\lambda = (\lambda_1,\ldots,\lambda_d) \;\text{with}\; \lambda_1 \ge \lambda_2 \ge \cdots \ge \lambda_d,\;\text{and}\;
    \sum_{i=1}^d \lambda_i = n.$ Here, $\mathcal{P}_{\lambda}$ and $\mathcal{Q}_{\lambda}$ denote the irreducible representation spaces associated with the symmetric group and the unitary group, respectively. Now, for each $\lambda \in Y^n_d$, we consider the projector onto the subspace 
$\mathcal{P}_{\lambda} \otimes \mathcal{Q}_{\lambda}$ given by
\begin{equation}
    \1^{(\mathcal{P}_{\lambda})}
    \otimes
    \1^{(\mathcal{Q}_{\lambda})},
\end{equation}
which makes $\{\1^{(\mathcal{P}_{\lambda})}
    \otimes
    \1^{(\mathcal{Q}_{\lambda})}\}_{\lambda \in Y^n_d}$
forms a valid projective measurement, since
\begin{equation} \forall \lambda \in Y^n_d \quad \1^{(\mathcal{P}_{\lambda})} \otimes \1^{(\mathcal{Q}_{\lambda})}\geq 0\quad\text{and}\quad\sum_{\lambda \in Y^n_d}\1^{(\mathcal{P}_{\lambda})} \otimes \1^{(\mathcal{Q}_{\lambda})} = \1. \end{equation}

Note that for each normalized partition $\bar{\lambda}=\frac{1}{n}(\lambda_1,\ldots,\lambda_d)$ arising from $\lambda$, the entropy of $\bar{\lambda}$ can be bounded as follows:
\begin{equation}
    0 \leq H(\bar{\lambda}) \leq \log d.
\end{equation}
We now divide the interval $[0,\log d]$ into $m$ bins with boundaries
\begin{equation}
    0 = h_0 \leq h_1 \leq \cdots \leq h_{m-1} \leq h_{m}=\log d,
\end{equation}
where the points $\{h_i\}_{i=1}^{m-1}$ are chosen independently according to a uniform distribution over the interval $[0, \log d]$.
Using such bins we introduce the following coarse-grained measurement %
\begin{equation}\label{cgm}
    M_j 
    :=
    \sum_{h_{j-1} \leq H(\bar\lambda) \leq h_j}
    \1^{(\mathcal{P}_{\lambda})}
    \otimes
    \1^{(\mathcal{Q}_{\lambda})} \quad\text{for}\quad j\in\{1,\ldots,m\}.
\end{equation}
Here, the measurement outcome directly determines the corresponding bin index. Then the entropy estimation protocol goes as follows:
\begin{enumerate}
    \item Measure $\rho^{\otimes n}$ using the measurement $\{M_j\}_{j=1}^m$.
    \item Upon obtaining an outcome $k$ (corresponding to the $k^{\text{th}}$ bin), we estimate the entropy of $\rho$ by choosing any value in the interval $[h_{k-1}, h_k]$.
\end{enumerate}
The above protocol enables estimation of the entropy of the unknown state $\rho$, while inducing only an asymptotically vanishing disturbance to the state. We rigorously establish this claim in the following theorem.

\begin{thm}[Error-disturbance trade-off for entropy estimation]
For $m = n^z$ with $0 < z < \tfrac{1}{2}$, the entropy estimation protocol fails with probability at most $\mathcal{O}\!\left(n^{\,z - \frac{1}{2}} (\log n)^{3/2}\right)$. Upon successful execution, the estimated entropy deviates from the true entropy $S(\rho)$ by at most $\mathcal{O}\!\left(n^{-z} \log n\right)$, while the disturbance induced on the state is at most $\mathcal{O}\!\left(n^{-p}\right)$.
\end{thm}
\begin{proof}
We proceed by identifying three possible cases that casues failure for the entropy estimation
\begin{enumerate}
\item \emph{Case 1: }Some bin boundary $h_k$ may lie close to the true entropy satisfying
\begin{equation}\label{ineq_entropy}
         |h_k-S(\rho)|<\frac{\sqrt{\log \left(n^{2p}\left(n+1\right)^{d(d+1)/2}\right)}}{\sqrt{n}}\log d+H_2\left(         \frac{\sqrt{\log \left(n^{2p}\left(n+1\right)^{d(d+1)/2}\right)}}{\sqrt{n}}\right):=a_n, 
\end{equation}
where $p$ is positive constant and $H_2(q)$ denotes the binary entropy defined as $H_2(q)=-\Big[q\log q+(1-q)\log(1-q)\Big]$. The occurrence of such a bin as the outcome of the measurement in Eq.~\eqref{cgm} may significantly disturb the state, since most of the weight of the state $\rho^{\otimes n}$ is concentrated in the support of the projector
\begin{equation}\label{interval1}
    \Pi_{a_n}:=\sum_{\lambda: |\bar\lambda-S(\rho)|\le a_n} \1^{(\mathcal{P}_{\lambda})}
    \otimes
    \1^{(\mathcal{Q}_{\lambda})},
\end{equation}
where $a_n$ given in Eq. \eqref{ineq_entropy}. This can be seen as follows:
\begin{align}
   \sum_{\lambda:|\bar\lambda-S(\rho)| \le a_n} &\!\!\!\Tr(\1^{(\mathcal{P}_{\lambda})}
    \otimes
    \1^{(\mathcal{Q}_{\lambda})}\rho^{\otimes n}) = \Pr\left\{\bar\lambda: |\bar\lambda-S(\rho)|\le a_n\right\} \nonumber\\
    &= \Pr\left\{\bar\lambda: |\bar\lambda-S(\rho)|\le \frac{\sqrt{\log \left(n^{2p}\left(n+1\right)^{d(d+1)/2}\right)}}{\sqrt{n}}\log d+H_2\left(         \frac{\sqrt{\log \left(n^{2p}\left(n+1\right)^{d(d+1)/2}\right)}}{\sqrt{n}}\right)\right\}\nonumber\\
    &\overset{(1)}{\geq} \Pr\left\{d_{\mathrm{TV}}(\bar\lambda,r)\le\frac{\sqrt{\log \left(n^{2p}\left(n+1\right)^{d(d+1)/2}\right)}}{\sqrt{n}}\right\}
    \overset{(2)}{\geq} 1-\left(n+1\right)^{d(d+1)/2}\exp\left(-n\left(\frac{\sqrt{\log \left(n^{2p}\left(n+1\right)^{d(d+1)/2}\right)}}{\sqrt{n}}\right)^2\right)\label{Fannes}\\=&1-n^{-2p},\label{Fannes2}        
\end{align}
where \(r\) denotes the decreasingly ordered spectrum of the state \(\rho\) which means $S(\rho)=H(r)$, and $d_{\text{TV}}(x,y)=\frac{1}{2}\sum_{i=1}^{d}|x_i-y_i|$ denotes the total variation distance. Note that any \(\bar{\lambda}\) satisfies
\begin{equation}
    d_{\mathrm{TV}}(\bar\lambda,r)
    \leq
    \frac{
        \sqrt{
            \log\!\left(
                n^{2p}(n+1)^{d(d+1)/2}
            \right)
        }
    }{\sqrt{n}}\quad\text{also satisfy}\quad|\bar\lambda-S(\rho)|\leq a_n,
\end{equation}
via Fannes-Audenaert inequality given in Proposition \ref{Fannes_Audenart_ineq_strict} which allows us to write first inequality in Eq.~\eqref{Fannes}. The second inequality in Eq.~\eqref{Fannes} follows from \cite{WrightOdonnel} that says:
For a mixed state $\rho$ with ordered spectrum $r = (r_1,\ldots,r_d)$, we have:
\begin{equation}
\forall\epsilon > 0,\quad \Pr\left\{\bar\lambda: d_{\mathrm{TV}}(\bar\lambda, r) > \epsilon\right\}
\le (n+1)^{\frac{d(d+1)}{2}} \exp(-2n\epsilon^2).
\end{equation}
The bin boundaries $\{h_i\}_{i=1}^{m-1}$ are chosen according to a uniform distribution over the interval $[0,\log d]$, therefore the probability of having a bin boundary in the interval given in Eq.~\eqref{interval1} is given as 
\begin{equation}
    \frac{1}{\log d
    } \left(2a_n\right)=2\left(\frac{\sqrt{\log \left(n^{2p}\left(n+1\right)^{d(d+1)/2}\right)}}{\sqrt{n}}+\frac{1}{\log d}\;H_2\left(         \frac{\sqrt{\log \left(n^{2p}\left(n+1\right)^{d(d+1)/2}\right)}}{\sqrt{n}}\right)\right).
\end{equation}
Therefore, the probability that one or more bin boundaries lie inside the interval $[S(\rho)-a_n, S(\rho)+a_n]$ used in Eq.~\eqref{interval1} can be upper bounded using the union bound:
\begin{align}
    &2m\left(
    \frac{\sqrt{\log \left(n^{2p}\left(n+1\right)^{d(d+1)/2}\right)}}{\sqrt{n}}
    +\frac{1}{\log d}\,
    H_2\!\left(
    \frac{\sqrt{\log \left(n^{2p}\left(n+1\right)^{d(d+1)/2}\right)}}{\sqrt{n}}
    \right)
    \right)\nonumber\\
    =&\;
    2n^{z-\frac{1}{2}}
    \left(
    \sqrt{\log \left(n^{2p}\left(n+1\right)^{d(d+1)/2}\right)}
    +\frac{\sqrt{n}}{\log d}
    H_2\!\left(
    \frac{\sqrt{\log \left(n^{2p}\left(n+1\right)^{d(d+1)/2}\right)}}{\sqrt{n}}
    \right)
    \right)=
    \mathcal{O}\!\left(
    n^{z-\frac{1}{2}}(\log n)^{3/2}
    \right),
\end{align}
where to write the first equality we use $m=n^z$. Therefore, the probability that at least one bin boundary lies within the interval in Eq.~\eqref{interval1} vanishes asymptotically whenever \(z<\tfrac12\).
    \item \emph{Case 2: } It may also happen that there exists no bin boundary \(h_i\) sufficiently close to \(S(\rho)\), i.e.,
\begin{equation}\label{entropy_ineq_case_2}
   \forall i, \quad |S(\rho)-h_i|
    >
    n^{-z}\log d
    \sqrt{
    \log \left(
    n^{2p}\left(n+1\right)^{d(d+1)/2}
    \right)
    }
    =
    \mathcal{O}\!\left(
    n^{-z}\log n
    \right).
\end{equation}
In such a situation, the estimated value of entropy may lie outside the allowed error tolerance. Since the bin boundaries $\{h_i\}_{i=1}^{m-1}$ are chosen independently and uniformly from $[0,\log d]$, the probability that all $h_i$ satisfies Eq.~\eqref{entropy_ineq_case_2} can be bounded from above as
\begin{align}
    &2\left(1-\frac{n^{-z}\log d\sqrt{\log \left(n^{2p}(n+1)^{d(d+1)/2}\right)}}{\log d}\right)^m = 2\left(1-n^{-z}\sqrt{\log \left(n^{2p}(n+1)^{d(d+1)/2}\right)}\right)^m \nonumber\\
    &\le 2\left(1-n^{-z}\sqrt{\left(2p+\frac{d(d+1)}{2}\right)\log n}\right)^m \le 2\exp\!\left(-m n^{-z}\sqrt{C\log n}\right)
    = \mathcal{O}\!\left(\exp\!\left(-\sqrt{\log n}\right)\right),
\end{align}
where we define $C := 2p + \frac{d(d+1)}{2}$. Here, the first inequality by using $n+1 \ge n$, the second from the bound $(1-x)^m \le e^{-mx}$, and the final step uses $m = n^z$.
    \item \emph{Case 3: } The coarse-grained measurement $\{M_j\}_{j=1}^m$ may yield an outcome corresponding to a bin that does not contain $S(\rho)$, resulting in an incorrect estimate of the entropy. Suppose that no bin lies within a distance $a_n$ of $S(\rho)$, where $a_n$ is defined in Eq.~\eqref{ineq_entropy} (i.e., case 1 did not happen). Then, by Eq.~\eqref{Fannes2}, the probability of such an event can be upper bounded as
\begin{equation}
    \Pr\left\{|\bar\lambda - S(\rho)| \ge a_n\right\}
    \le n^{-2p}.
\end{equation}
Therefore, the total failure probability is dominated by Case 1, and hence the overall failure probability scales as 
\begin{equation}\label{entropy_estimation_avg_bd}
    \mathcal{O}\!\left(n^{\,z-\frac{1}{2}}(\log n)^{3/2}\right).
\end{equation}

\end{enumerate}
We say that the entropy estimation is successful whenever none of Cases~1, 2, or~3 occurs. In this situation there exist an index $i$ such that
\begin{equation}\label{above_qen}
|S(\rho)-h_i|
\leq n^{-z}\log d
\sqrt{
\log \left(
n^{2p}(n+1)^{d(d+1)/2}
\right)
}
= \mathcal{O}\left(n^{-z}\log n\right),
\end{equation}
 follows from the fact that Case 2 does not occur hence we obtain Eq. \eqref{above_qen} by taking negation of the statement given in Eq.~\eqref{entropy_ineq_case_2}.
 Suppose we perform the measurement described in Eq.~\eqref{cgm} and that the entropy estimation is successful. This means that the bin containing $S(\rho)$ occurs (i.e., Case~3 does not occur) and, moreover, that there is no bin boundary within the interval specified in Eq.~\eqref{interval1} (i.e., Case~1 does not occur). From Eq.~\eqref{Fannes2}, it follows that the bin containing $S(\rho)$ occurs with probability at least $1-n^{-2p}$. Therefore, by the Gentle Operator Lemma, the disturbance caused by the measurement is small, and the post-measurement state $\tilde{\rho}$ satisfies
\begin{equation}
    \|\rho^{\otimes n}-\tilde{\rho}\|_{1} \leq \sqrt{2}\,n^{-p}\qquad\text{where}\qquad \tilde{\rho}=\frac{\Pi_{a_n}\rho^{\otimes n}\Pi_{a_n}}{\Tr\left(\Pi_{a_n}\rho^{\otimes n}\Pi_{a_n}\right)}.
\end{equation}
\end{proof}

\subsection{Estimation of the average charges of an unknown state}
In this section, we aim to efficiently estimate the expectation values of the charges $\Tr(L_c\rho)$, for $c \in \{1,\ldots,t\}$ while causing an asymptotically vanishing disturbance to the state. Since all the charges commute, there exists a common eigenbasis $\{ \ket{x} \}_{x=1}^d$ such that
\begin{equation}
    \forall c \in \{1,\ldots,t\}, \qquad
    L_c = \sum_{x=1}^{d} L_c(x)\, \ket{x}\!\bra{x}.
\end{equation}
Thus, for any $c\in\{1,\ldots,t\}$ we have $ \Tr(L_c\rho)=\sum_{x=1}^d L_c(x)\langle x|\rho|x\rangle$. Therefore, to estimate $\Tr(L_c\rho)$, it suffices to estimate the diagonal matrix elements of $\rho$ in the common eigenbasis $\{\ket{x}\}_{x=1}^d$, that is, the quantities
\begin{equation}
    \forall x\in \{1,\ldots,d\}, \qquad \langle x|\rho|x\rangle:=\alpha_x \quad\text{which allow us to obtain}\quad \Tr(L_c\rho)=\sum_{x=1}^d L_c(x)\langle x|\rho|x\rangle = \sum_{x=1}^d L_c(x)\alpha_x.
\end{equation}
We proceed analogously to the previous case by constructing a coarse-grained measurement from the projective measurement that would estimate $\alpha_x$ exactly, but at the expense of causing significant disturbance to the state. We employ the coarse-grained version of the measurement, which avoids significant disturbance to the state $\rho^{\otimes n}$ while introducing only an estimation error that vanishes asymptotically.

If one is unconcerned with disturbing the state, then estimating $\langle x|\rho|x\rangle$ for a fixed basis vector $|x\rangle$ is straightforward: one may simply perform the measurement $\{\ketbra{x},\1-\ketbra{x}\}$ independently on each copy of $\rho$. Then the number of occurrences of the outcome $\ketbra{x}$ is binomially distributed with mean $n\alpha_x$ and variance $n\alpha_x(1-\alpha_x)$. Consequently, $\alpha_x$ can be estimated with accuracy $\mathcal{O}(n^{-1/2})$. However, such a measurement can severely disturb the state. For example, when applied to the state $\frac{1}{\sqrt{2}}\left(\ket{x}+\ket{x^\perp}\right)$, the measurement drastically alters the coherence between the two components. 

Rather than measuring each copy individually, the same procedure can equivalently be expressed as a collective measurement acting jointly on all $n$ copies. The corresponding measurement operators are
\begin{equation}\label{energy_mmt_fine_grained}
    N_k=\sum_{y\in\{0,1\}^n,\ |y|=k}\ \bigotimes_{j=1}^n \left(y_j\ketbra{x}+(1-y_j)\left(\1-\ketbra{x}\right)\right),
\end{equation}
where $k\in\{0,\dots,n\}$ and $|y|$ denotes the Hamming weight of the bit string $y$, i.e., the number of $1$'s in $y$. We emphasize that the measurement described in Eq.~\eqref{energy_mmt_fine_grained} can significantly disturb the state, and to avoid that, we will introduce the coarse-grained measurement as earlier. 

We now partition the interval $[0,n]$ into $m$ bins by selecting the bin boundaries $\{b_i\}_{i=1}^{m-1}$ independently according to the uniform distribution on $[0,n]$. The coarse-grained measurement $\{N_j\}_{j=1}^m$ is given as follows:
\begin{equation}
    N'_j = \sum_{b_{j-1}\leq k<b_j} N_k, \quad\text{where}\quad j\in\{1,\ldots, m\}.
\end{equation}
Like previous scenario of entropy estimation, the measurement outcome directly determines the corresponding bin index. Then the  estimation protocol for $\alpha_x$ goes as follows:
\begin{enumerate}
    \item Measure $\rho^{\otimes n}$ using the measurement $\{N'_j\}_{j=1}^m$.
    \item Upon obtaining an outcome $k$ (corresponding to the $k^{\text{th}}$ bin), we estimate $\alpha_x$ by choosing any value in the interval $[b_{k-1}, b_k]$.
\end{enumerate}

Now, we will directly use a result from \cite{BHL2006_PRA} to ensure successful estimation of the quantity $\alpha_x$ while causing asymptotically vanishing disturbance to the state. 

\begin{thm}[{Error disturbance trade-off for average charge estimation \cite[Prop.~1]{BHL2006_PRA}}]
   For \(m=n^z\) with \(0<z<\frac12\), the protocol for estimating \(\alpha_x\) has failure probability at most $\mathcal{O}\!\left(n^{\,z-\frac12}\log n\right).$ Upon successful execution, the estimated value of $\alpha_x$ deviates from true value of $\alpha_x$ by at most $\mathcal{O}\!\left(n^{-z} \log n\right)$ while the disturbance induced on the state is upper-bounded by $\mathcal{O}\!\left(n^{-p}\right)$ where $p$ is a positive constant.
\end{thm}

By the union bound, the probability of failure to estimate all the \(\alpha_x\) can be bounded by summing the individual failure probabilities over \(x \in \{1,\ldots,d\}\). This yields an overall failure probability upper bounded by
\begin{equation}\label{charge_estimation_avg_bd}
    \mathcal{O}\!\left(n^{\,z-\frac12}\log n\right).
\end{equation}

\subsection{The total failure probability and the disturbance induced on the state due to the estimation of entropy and average charges}

In this section, we compute the total failure probability of the estimation protocol, which consists of contributions from both entropy estimation and average charge estimation. In particular, we have
\begin{align}
    &\text{Failure probability of entropy estimation}
    + \text{Failure probability of average charge estimation for all charges} \nonumber \\
    &\phantom{ail}\leq
    \mathcal{O}\!\left(n^{\,z-\frac{1}{2}}(\log n)^{3/2}\right)
    +
    \mathcal{O}\!\left(n^{\,z-\frac{1}{2}}\log n\right) =
    \mathcal{O}\!\left(n^{\,z-\frac{1}{2}}(\log n)^{3/2}\right),
\end{align}
where the inequality follows from Eq.~\eqref{entropy_estimation_avg_bd} and Eq.~\eqref{charge_estimation_avg_bd}, and the final scaling is dominated by the first term. The total disturbance induced by the sequential entropy and average-charge estimation measurements is controlled using the non-commutative union bound (see Proposition~\ref{NCUB}), followed by the Gentle Operator Lemma~\ref{Gentle_op_lem}. Since the total failure probability is

$$
\mathcal{O}\!\left(n^{\,z-\frac12}(\log n)^{3/2}\right),
$$

and the number of measurements is constant, the gentle Operator Lemma yields an overall disturbance of

$$
\mathcal{O}\!\left(n^{\frac{z}{2}-\frac14}(\log n)^{3/4}\right).
$$


\section{Useful results}
Here we collect well-known tools that we use elsewhere in the text. 
\begin{prop}[Fannes-Audenaert inequality \cite{Fannes1973,Audenaert2007,Winter2016, Chiribella_Renner}]\label{Fannes_Audenart_ineq_strict}
    For two quantum states $\rho$ and $\sigma$ in $d$-dimensional Hilbert space such that $\frac{1}{2}\|\rho-\sigma\|_{1} \leq \delta\leq1$, we have
    \begin{equation}
        |S(\rho)-S(\sigma)|\leq \delta\log(d)+H_2(\delta):=\xi(\delta),
    \end{equation}
    where $H_2(\delta)$ denotes the binary entropy defined as $H_2(\delta)=-\delta\log(\delta)-(1-\delta)\log (1-\delta)$.
\end{prop}
\begin{thm}[Hoeffding's inequality \cite{Hoeffding_ineq}]
\label{Hoeffding_ineq}
    Let $X_1,\;X_2,\;\ldots,\;X_n$ are independent and bounded random variables with $X_i\in[a,b]$ for all $i$, where $a$ and $b$ are finite. Then, for any $f>0$,
\begin{equation}
   \Pr\left\{\left|\frac{1}{n}\sum_{i=1}^{n}X_i - \mathbb{E}[X]\right| \geq f\right\} \leq \mathrm{exp}\Bigg(-\frac{2nf^2}{(b-a)^2}\Bigg).
\end{equation}
\end{thm}

\begin{lem}[Gentle operator lemma \cite{wilde_2017,Winter_Gentle_op_lem}]
\label{Gentle_op_lem}
    Consider a density operator $\rho$ and a measurement operator $\Lambda$ where $0\leq \Lambda \leq I$. The measurement operator could be an element of a POVM. Suppose that the measurement operator $\Lambda$ has a high probability of detecting state $\rho$ i.e.,
    \begin{equation}
        \Tr(\Lambda\rho)\geq 1-\varepsilon,
    \end{equation}
    where $\epsilon\in[0,1]$. Then, the post-measurement state 
    \begin{equation}
        \rho' = \frac{\sqrt{\Lambda}\rho\sqrt\Lambda}{\Tr(\Lambda\rho)}
    \end{equation}
    satisfies
    \begin{equation}
        \|\rho-\rho'\|_{1}\leq 2\sqrt{\varepsilon}.
    \end{equation}
\end{lem}

\begin{thm}[Non-commutative union bound \cite{GaoPRA}]
\label{NCUB}
Let \(\sigma\) be a quantum state and  \(\Pi_1,\ldots,\Pi_m\) be projectors. Then
\begin{equation}
    \Bigg\|\sigma-\frac{\Pi_m\cdots \Pi_1 \sigma \Pi_1 \cdots \Pi_m}{\Tr\left(\Pi_m\cdots \Pi_1 \sigma \Pi_1 \cdots \Pi_m\right)}\Bigg\|_1\leq 2\sqrt{
\sum_{i=1}^m
\Tr\left((\1-\Pi_i)\sigma\right)
},
\end{equation}
and
\begin{equation}
\Tr\left(
\Pi_m\cdots \Pi_1 \sigma \Pi_1 \cdots \Pi_m
\right)
\geq 1 - 4\left(\sum_{i=1}^m \Tr\left((\1-\Pi_i)\sigma\right)\right).
\end{equation}
\end{thm}

\end{document}